\RequirePackage{fix-cm}
\documentclass{vldb-no-cprt}          
\newtheorem{theorem}{Theorem}
\newtheorem{definition}[theorem]{Definition}
\usepackage{times}
\usepackage{mathptmx}      
\usepackage{graphicx}
\usepackage{flushend}
\usepackage{subcaption}
\usepackage{microtype}
\usepackage{ifthen}
\usepackage{epstopdf}
\usepackage[usenames,dvipsnames]{xcolor}
\usepackage{listings}

\newcommand\thepapertitle{Efficiently making (almost) any concurrency control mechanism serializable}
\newcommand\thepaperkeywords{Concurrency control, serializability, serial safety net, read optimization, read-mostly transactions, ERMIA, main-memory databases}
\usepackage{float}
\usepackage[
  setpagesize=false,
  pageanchor=false,
  plainpages=false,
  pdfpagelabels=false,
  bookmarks,
  bookmarksnumbered,
  pdfborder={0 0 0},  
  colorlinks=true,
  linkcolor=blue,
  citecolor=blue,
  pdftitle={\thepapertitle},
  pdfauthor={Tianzheng Wang},
  pdfkeywords={\thepaperkeywords},
]{hyperref}

\floatstyle{ruled}
\usepackage{algorithm}
\usepackage[noend]{algorithmic}
\usepackage{cite}
\usepackage{mathtools}
\usepackage{centernot}
\usepackage{balance}
\newcommand{\rdep}{\xleftarrow{\mathit{b}}}
\newcommand{\rdepstar}{\xleftarrow{\mathit{b*}}}
\newcommand{\fdep}{\xleftarrow{\mathit{f}}}

\newcommand{\xdep}{\leftarrow}
\newcommand{\txdepimpl}[3]{\xleftarrow{\mathit{#1{:}#2#3}}}
\newcommand{\txdep}[3][]{%
  \ifthenelse{\equal{#1}{}}
  {\txdepimpl{#2}{}{#3}}
  {\txdepimpl{#2}{#1{:}}{#3}}
}
\newcommand{\rwdep}[1][]{\txdep[#1]{r}{w}}
\newcommand{\wrdep}[1][]{\txdep[#1]{w}{r}}
\newcommand{\wwdep}[1][]{\txdep[#1]{w}{w}}
\newcommand{\wxdep}[1][]{\txdep[#1]{w}{x}}

\begin{document}
\title{Efficiently making (almost) any\\concurrency control mechanism serializable}
\author{
\alignauthor
$
\begin{array}{cccc}
  \mbox{Tianzheng Wang} & \mbox{Ryan Johnson} & \mbox{Alan Fekete} & \mbox{Ippokratis Pandis}\\
  \mbox{\affaddr{University of Toronto}} &
  \mbox{\affaddr{LogicBlox}} &
  \mbox{\affaddr{University of Sydney}} &
  \mbox{\affaddr{Amazon Web Services}}\\
  \mbox{\sf tzwang@cs.toronto.edu} &
  \mbox{\sf ryan.johnson@logicblox.com} &
  \mbox{\sf alan.fekete@sydney.edu.au} &
  \mbox{\small\sf ippo@amazon.com}
\end{array}
$
}

\maketitle

\begin{abstract}
Concurrency control (CC) algorithms must trade off strictness for performance.
In particular, serializable CC schemes generally pay higher cost to prevent
anomalies, both in runtime overhead such as the maintenance of lock tables, and
in efforts wasted by aborting transactions. We propose the serial safety net
(SSN), a serializability-enforcing certifier which can be applied on top of
various CC schemes that offer higher performance but admit anomalies, such as
snapshot isolation and read committed.  The underlying CC mechanism retains
control of scheduling and transactional accesses, while SSN tracks the resulting
dependencies. At commit time, SSN performs a validation test by examining only
\textit{direct} dependencies of the committing transaction to determine whether
it can commit safely or must abort to avoid a potential dependency cycle.

SSN performs robustly for a variety of workloads. It maintains the
characteristics of the underlying CC without biasing toward a certain type of
transactions, though the underlying CC scheme might.  Besides traditional OLTP
workloads, SSN also efficiently handles heterogeneous workloads which include a
significant portion of long, \textit{read-mostly} transactions. SSN can avoid
tracking the vast majority of reads (thus reducing the overhead of
serializability certification) and still produce serializable executions with
little overhead. The dependency tracking and validation tests can be done
efficiently, fully parallel and latch-free, for multi-version systems on modern
hardware with substantial core count and large main memory.

We demonstrate the efficiency, accuracy and robustness of SSN using extensive
simulations and an implementation that overlays snapshot isolation in ERMIA, a
memory-optimized OLTP engine that supports multiple CC schemes. Evaluation
results confirm that SSN is a promising approach to serializability with robust
performance and low overhead for various workloads.
\end{abstract}

\section{Introduction}
Concurrency control (CC) algorithms interleave read and write requests from 
multiple users simultaneously, while giving the (perhaps imperfect) illusion 
that each transaction has exclusive access to the data. Serializable CC 
mechanisms generate concurrent transaction executions that are equivalent to 
\textit{some} serial ones.
This is desirable for users, because serializable executions
never have anomalies (e.g., lost update and write skew) and can preserve 
integrity constraints over the data. Enforcing a cycle-free transaction 
dependency graph is a necessary condition to achieve serializability, 
and is the focus of this work.\footnote{Phantom protection,
as we will discuss later in Section~\ref{sec:phantoms}, is another necessary, 
but largely orthogonal issue.}
Some CC schemes---such as two-phase locking (2PL) and serializable snapshot 
isolation (SSI)~\cite{crf09}---forbid all dependency cycles to guarantee serializability, 
but in doing so they also forbid many valid serializable schedules.

Traditional serializable CC schemes have been either pessimistic or optimistic.
In today's environment of massively parallel, large main memory hardware, 
it is common for the working set---or even the whole database---to fit in main memory. 
I/O operations are completely out of the critical path. 
Existing pessimistic scheme implementations often scale poorly in this situation, 
due to physical contention (e.g., on centralized lock tables~\cite{JohnsonPHAF09,PandisJHA10}). 
Lightweight optimistic concurrency control (OCC)~\cite{KungR81} is favored in 
many recent memory-optimized systems \cite{FOEDUS,hekaton,silo}, 
but OCC is known to be unfriendly to heterogeneous workloads that 
have a significant amount of analytical operations and \textit{read-mostly}
transactions \cite{RobustCC,ERMIA}.
Considering the performance impact of both kinds of serializable CC, 
many designs have non-serializable execution as the common case.
For example, although serializable SI (SSI)~\cite{crf09} has 
been implemented in PostgreSQL to ensure full serializability \cite{pg12}, 
Read Committed (RC) is still the default isolation level in PostgreSQL for 
performance reasons \cite{Pgsql-Isolation}, 
and a similar default is found in most widely-used database systems.
Sometimes there is no available isolation level that guarantees serializability. 
Whenever an application uses transactions that may not be serializable, 
data corruption is a risk,
so our focus is on guaranteeing serializability while reducing the
performance degradation as much as possible.

\begin{figure}
\centering
\includegraphics[width=0.98\columnwidth]{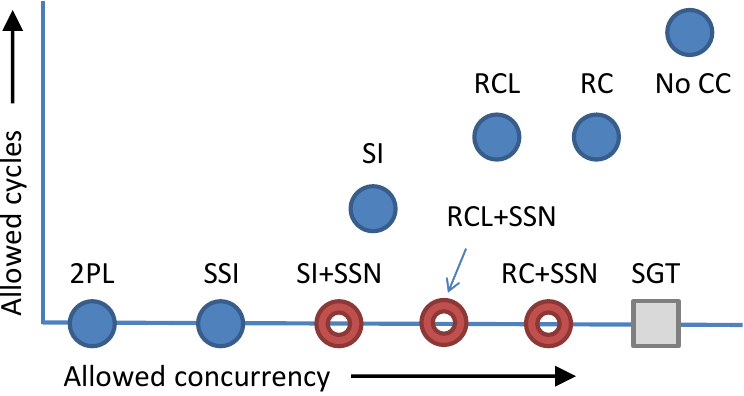}
\caption{Relative merits of existing CC schemes (solid dots) vs. the serial 
safety net (hollow dots).}
\label{fig:signature-chart}
\end{figure}

Fig.~\ref{fig:signature-chart} illustrates the relative strictness vs. 
performance trade-off for several well-known CC schemes. At one extreme, strict 
2PL ensures serializability but offers low concurrency because readers and 
writers block each other. At the other extreme, a system with no CC whatsoever 
(No CC) offers maximum concurrency but admits often intolerable anomalies 
(e.g., dirty reads and lost writes). With low performance cost, RC and its 
lock-based variant (RCL) offer much stronger semantics than No CC, and are 
often used in practice. Snapshot isolation (SI) makes a very attractive 
compromise, offering reasonably strict semantics and fairly high performance, 
while SSI offers full serializability but lowers concurrency. 
Fully precise serialization graph testing (SGT)~\cite{CasanovaB80} allows all 
(and only) cycle-free executions, but is impractical as every commit requires 
an expensive search for cycles over an ever-changing dependency graph.

\subsection{The serial safety net in a nutshell}
This paper describes the serial safety net (SSN), an efficient, general-purpose 
certifier to enforce serializability on top of a variety of CC schemes, such as 
RC and SI. SSN does not dictate access patterns---the underlying CC scheme does 
that---but instead tracks dependencies and aborts transactions that might close 
a dependency cycle if allowed to commit. 
SSN admits false positives, but is much more accurate than 
prior practical serializable CC schemes (e.g., 2PL and SSI). As illustrated 
by Fig.~\ref{fig:signature-chart}, SSN guarantees serializability with 
concurrency levels not drastically worse than the underlying CC scheme.
In particular, SI+SSN, RC+SSN, and RCL+SSN all allow higher
concurrency than 2PL and SSI.
The SSN infrastructure can 
also be used to prevent phantoms, and so offers full protection for any CC 
mechanism that forbids dirty reads and lost writes. The majority of schemes 
meet these constraints; read uncommitted (ANSI SQL) and eventual 
consistency (NoSQL favorite) are perhaps the most notable exceptions, by allowing 
dirty reads and lost writes, respectively. SSN thus expands significantly the 
universe of CC schemes that can be made serializable.

SSN can be implemented in both traditional disk-based systems and recent
main-memory databases. We focus on multi-version main-memory systems in this
paper. To facilitate dependency tracking, SSN requires globally unique
timestamps, which can be generated from a centralized source (e.g., a counter
incremented by the atomic \texttt{fetch-and-add} instruction on
x86~\cite{IntelManual}) or by augmenting unique thread-local counters to
block-allocated timestamps from a centralized source.

The gist of SSN consists of two parts: (1) a low watermark 
$\pi(T)$ for transaction $T$ that summarizes ``dangerous'' transactions that 
committed before $T$ but which must be serialized after $T$, and (2) a conservative 
validation test that is applied when $T$ commits at time $c(T)$: if $U$ has 
already committed and had a conflict with $T$ (i.e., $U$ must be serialized 
\textit{before} $T$), then $\pi(T)\nobreak \leq\nobreak c(U)\nobreak <\nobreak c(T)$ 
is forbidden because $U$ might also need to be serialized \textit{after} $T$, 
forming a cycle in the dependency graph. We prove that maintaining this 
\textit{exclusion window} suffices to prevent all cycles in the serial dependency graph, 
and then show how phantoms can be converted into dependency cycles so that SSN 
can enforce truly serializable executions in systems that otherwise 
lack phantom protection. We also show that $\pi(T)$ can be
computed efficiently for multi-version systems.

One unique aspect of SSN is that it works in spite of bugs, omissions, or 
unanticipated behaviors in the underlying CC scheme, so long as the basic 
requirements still hold. This protection is important, because CC schemes tend 
to be complex to implement, and bugs can lead to subtle problems that are 
difficult to detect and reproduce. Unanticipated behaviors are even more 
problematic. For example, a read-only anomaly in SI arises only if a reader 
arrives at exactly the wrong moment \cite{fekete-read-only-anomaly}. This 
anomaly was not discovered until SI had been in use for many years. Assuming 
SSN is implemented correctly---hopefully achievable, given its 
simplicity---bugs or unexpected behaviors in the CC scheme that would confuse 
applications, will instead trigger extra transaction aborts caused by SSN. The 
application sees only serializable executions that preserve data integrity. 

SSN is amenable to a variety of workloads and does not exaggerate the
underlying CC's favor for either reader or writer accesses. Moreover,
SSN's efficient dependency tracking and exclusion window test give
the opportunity to optimize emerging heterogeneous workloads that contain 
a significant portion of long, \textit{read-mostly} transactions: reads of stale
records that are not updated recently do not have to be tracked in the
transaction's read set. This greatly reduces bookkeeping footprint and
improves performance.

\subsection{Contributions and paper organization}
We have introduced the main techniques of SSN in an earlier paper
\cite{SSN-DAMON}. SSN uses only local knowledge of each transaction and its
direct conflicts to determine whether committing a transaction will close a
potential dependency cycle. In this paper, we leverage these main techniques to
further propose a generic approach to optimizing emerging heterogeneous
workloads, and an efficient parallel commit protocol with minimum overhead for
today's memory-optimized, multi-version systems.

Compared to the earlier version of this paper, we evaluate SSN with a much wider
set of experiments, both in simulation and ERMIA~\cite{ERMIA}, a recent OLTP
system optimized for massively parallel processors and large main memory.
Simulation results show that SSN works well under a wide variety of
circumstances, including both lock-based and multi-version CC, mixed workloads
and very high contention. Evaluation using ERMIA on a quad-socket, 60-core Xeon
server shows that SSN scales as well as the underlying CC scheme. In
particular, SSN's optimization for read-mostly transactions can significantly
reduce last level cache misses and perform more than 2$\times$ better than an
efficient parallel SSN implementation without the optimization. Compared to SSI,
SSN matches its performance for workloads with low and medium contention that do
not stress the CC protocol. For high-contention workloads with retrying aborted
transactions, SSN can provide more robust performance and better accuracy for
both write-intensive and read-only transactions.

The rest of the paper is organized as follows. In Sect.~\ref{sec:background}, we
give background on serial dependency graphs that we use throughout the paper to
understand serializability properties. Sect.~\ref{sec:ssn} discusses the design
and presents a theoretical proof of the correctness of SSN.
Sect.~\ref{sec:implementation} gives an efficient and scalable
implementation of SSN for multi-version systems, leveraging parallel programming
techniques. In Sect.~\ref{sec:reduce-overhead}, we discuss ways of making SSN
lightweight and efficient, including how we optimize read-only and heterogeneous
workloads using SSN. Sect.~\ref{sec:phantoms} extends the SSN infrastructure to
prevent phantoms for systems that are not otherwise phantom-free. We then
present evaluation results of SSN in Sect.~\ref{sec:simulations}
and~\ref{sec:eval}, using simulation and implementation respectively. We survey
related work in Sect.~\ref{sec:related-work} and conclude in
Sect.~\ref{sec:conclusion}.

\section{Serial dependency graphs}
\label{sec:background}
We model the database as a multi-version system that consists
of a set of records \cite{adya99}. Each transaction consists of a sequence of 
reads and writes, each dealing with a single record. 
In this model, each record is seen as a totally-ordered sequence of
\textit{versions}. A write always generates a new version at the end of 
the record's sequence; a read returns a version in the record's sequence
that the underlying CC mechanism deems appropriate.
In the model, each record exists forever, with a continually growing set 
of versions. In practice, obsolete versions that are no longer needed 
are periodically recycled to avoid wasting storage space.
Insertions and deletions are represented using a special ``invalid'' value, 
for the initial version of a record that has not yet been inserted, 
and also for the last version of a deleted record.
Insertions are updates that replace invalid versions. 
A deletion flags a record as invalid without physically deleting it, 
and the record can continue to participate in CC if needed. 
The physical deletion is performed in background once the record 
is no longer reachable~\cite{graefe10}.
In this model, we do not explicitly model the case where
a transaction reads a record it has previously written, 
because doing so does not add new edges to the dependency graph, 
i.e., no new cycles can arise.
Many real systems ensure that a read will return the version that the 
transaction itself wrote. 
We note, however, that there are exceptions: 
certain OCC-based systems~\cite{FOEDUS,silo} 
do not allow a transaction to read its own writes.

We first only consider the serial dependency cycles that may arise 
among individual records that are read and written. In the absence of 
insertions, preventing such cycles produces a serializable schedule. 
In Sect.~\ref{sec:phantoms}, we extend these concepts to include analogs of 
hierarchical locking, lock escalation, and predicate-based selection to
prevent phantoms. 

Accesses by transaction $T$ generate {\em serial dependencies} that constrain 
$T$'s place in the global partial order of transactions. Serial dependencies 
can take two forms:
\begin{enumerate}
\item $T_i \wxdep{} T$ (read/write dependency): $T$ read ($T_i \wrdep{} T$) or
overwrote ($T_i \wwdep{} T$) a version 
that $T_i$ created, so $T$ must be serialized after $T_i$.
\item $T \rwdep{} T_j$ (read anti-dependency): $T$ read a version that $T_j$ 
overwrote, so $T$ must be serialized before $T_j$.
\end{enumerate}

A read implies a dependency on the transaction that created the returned 
version, and an anti-dependency from the transaction that (eventually) 
produces the next version of the same record (overwriting the version that 
was read). A write implies a dependency on the transaction that generated 
the overwritten version as well as dependencies on all reads that access the
new version.
Accessing different versions of the same record (e.g., a non-repeatable read) 
within a transaction implies a serialization failure: 
$T_1 \rwdep{} T_2 \wrdep{} T_1$. 

We use $T \xdep U$ to represent a serial dependency of either case: 
either $T \wxdep U$ or $T \rwdep U$, and we say that $T$ is a 
direct predecessor of $U$ (i.e., $U$ is a direct successor of $T$). 
Note that in the former case $x$ can be $r$ or $w$.
The set of all serial dependencies between committed transactions 
forms the edges in a directed graph $G$, whose vertices are committed 
transactions and whose edges indicate required serialization ordering 
relationships. When a transaction commits, it is added to $G$, along 
with any edges involving previously committed transactions. $T$ may 
also have {\em potential edges} to uncommitted dependencies, which 
will be added to $G$ if/when those transactions commit. 

Note that our notation puts the arrowhead of a dependency arrow near 
the transaction that must be serialized before the other. This is 
the reverse of the usual notation \cite{adya99} but it makes 
the arrowhead look similar to the transitive effective ordering 
relation symbol we define next.

We define a relation $\prec$ for $G$, such that 
$T_i\nobreak \prec\nobreak T_j$ means $T_i$ is ordered before 
$T_j$ along some path through $G$ (i.e., $T_i \xdep \ldots \xdep T_j$).
We say that $T_i$ is a predecessor of $T_j$ (or equivalently, 
that $T_j$ is a successor of $T_i$). When considering 
\textit{potential edges}, we can also speak of potential successors 
and predecessors. These are transactions for which the potential edges 
(along with edges already in $G$) require them to be serialized after 
(or respectively before) $T$.

A cycle in $G$ produces 
$T_i \prec T_j \prec T_i$, and indicates a serialization failure because 
$G$ then admits no total ordering. 
The simplest cycles involve two transactions and two edges:
\begin{enumerate}
\item $T1 \wxdep{} T2 \wxdep{} T1$. $T1$ and $T2$ saw each 
others' writes (isolation failure).
\item $T1 \wxdep{} T2 \rwdep{} T1$. $T2$ saw some, but not all, 
of $T1$'s writes (atomicity failure).
\item $T1 \rwdep{} T2 \rwdep{} T1$. $T1$ and $T2$ each overwrote a value 
that the other read (write skew).
\end{enumerate}

In our work, a central concept is the relationship between the partial order
of transactions that $G$ defines, and the total order defined by their commit 
times. At the moment transaction $T$ enters pre-commit, we take a monotonically
increasing timestamp, and call it $c(T)$. An edge in $G$ is a \textit{forward edge} 
when the predecessor committed first in time, and a \textit{back edge} when the 
\textit{successor} committed first. A forward edge can be any type of 
dependency, but (for the types of CC algorithms we deal with, which enforce write 
isolation) back edges are always read anti-dependencies where the overwrite 
committed before the read. We denote forward and back edges as 
$T_1 \fdep T_2$ and $T_1 \rdep T_2$, respectively. Let us write 
$T_0 \rdepstar T_k$ for the reflexive and transitive back edge situation 
where $T_0$ is reachable from $T_k$ without following any forward edges, 
e.g., $T_0 \rdep T_1 \rdep T_2 \rdep T_3 \ldots \rdep T_{k-1} \rdep T_k$.
Note that $T \rdepstar T$ always holds.

We next describe a representative but not exhaustive sampling of CC mechanisms
that will be used for both discussions and evaluations in the rest of the paper:
\begin{itemize}
\item {\bf Read Committed (RC).} Reads return the newest committed version 
of a record and never block. Writes add a new version that overwrites the 
latest one, blocking only if the latter is uncommitted. Allows dependency 
cycles but forbids isolation failures (dirty reads and lost writes).

\item {\bf Read Committed with Locking (RCL)}. An RC variant (with the same 
types of cycles) that can be implemented with a single-version system using 
in-place updates. RCL is typically achieved by combining short-duration read 
locks with long-duration write locks. Readers and writers alike must block 
until the latest version is committed, but readers do not block writers. 

\item {\bf Snapshot Isolation (SI).} Each transaction reads from a consistent 
\textit{snapshot}, consisting of the newest version of each record that predates 
some timestamp (typically, the transaction's start time). Writers must abort if 
they would overwrite a version created after their snapshot. Allows write skew 
anomalies, but forbids isolation failures and enforces write atomicity.

\item {\bf Serializable Snapshot Isolation (SSI).} 
Like SI, but forbids the ``dangerous structure'': 
$T_1 \rwdep{} T_2 \rwdep{} T_3$ where $T_3$ committed first \cite{crf09} 
(with some exceptions made for read-only transactions \cite{pg12}). 
No cycles are possible and so all executions are serializable.

\item {\bf Strict Two-Phase Locking (2PL).} Used by many single-version 
systems with long-duration read and write locks. Reads return the newest 
version of a record, blocking if it has not committed yet. Writes replace 
the latest version, blocking if there are any in-flight reads or writes 
on the record by other transactions. No cycles are possible.

\end{itemize}

SSN can work with most realistic CC schemes that are at least as strong 
as RC (formal requirements are given in Sect.~\ref{sec:ssn}). 
We are especially interested in weaker CC schemes that allow 
atomicity failures, non-repeatable reads, write skew, and more complex 
cycles in $G$, including various forms of read skew 
(e.g. $T_1 \xdep{} T_3 \rwdep T_2 \xdep T_4 \rwdep T_1$). 

\begin{figure*}[t]
\includegraphics[width=\textwidth]{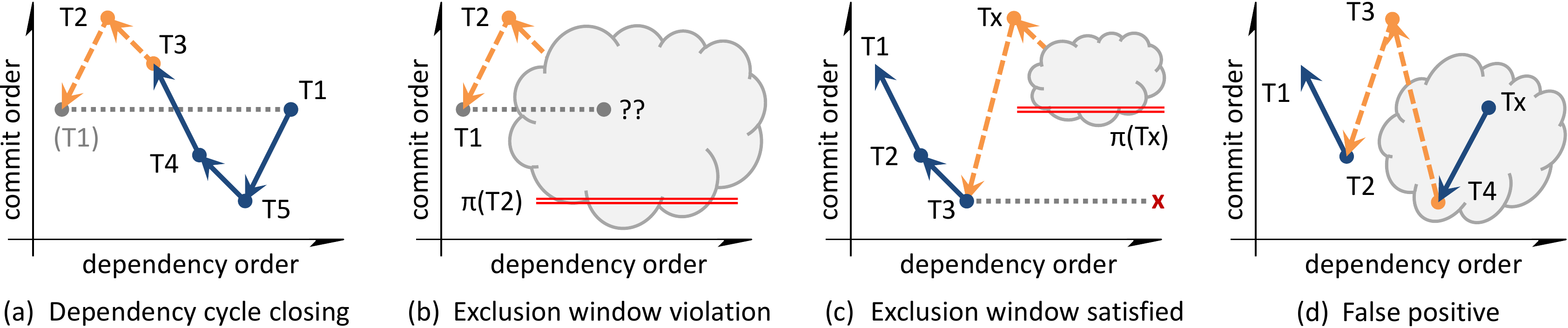}
\caption{A pictorial motivation and description of SSN.
Subsets of the dependency graph are 
shown in a serial-temporal layout where forward and back edges always have 
positive and negative slopes, respectively.}
\label{fig:ssn-in-pics}
\end{figure*}

\section{SSN: The serial safety net}
\label{sec:ssn}
In this section, we first describe how SSN prevents committing transactions 
that will close potential cycles in the dependency graph. We then formally prove the 
correctness of SSN and compare it with other serializable CC schemes.

\subsection{Preventing dependency cycles}
Given a CC scheme that admits cycles in the serial dependency graph, 
SSN can be layered on top as a pre-commit protocol to abort transactions 
that might form potential cycles if committed. Although SSN can be overlaid on 
various CC schemes, we require the underlying CC scheme forbid lost 
writes and dirty reads (unless it is the transaction reading its own 
writes), which is effectively as strong as RC.

In addition to the commit timestamp $c(T)$ of transaction $T$, 
SSN associates $T$ with two other timestamps: $\pi(T)$ and $\eta(T)$,
which are respectively the low and high watermarks used to detect 
conditions that might indicate a cycle in the dependency graph $G$ if $T$
is committed.
We define $\pi(T)$ as the commit time of $T$'s oldest successor 
$U$ reached through a path of back edges:
\begin{align*}
 \hspace*{1.2cm} \pi(T) & = min\left(c(U) : T \rdepstar U\right) \\
  &= min\left(\left\{\pi(U) : T \rdep U\right\} \cup \left\{c(T)\right\}\right)
\end{align*}
The first equation captures the definition, in which $T$'s successor $U$ 
that overwrote versions read by $T$, committed first, forming a back 
edge that represents a read anti-dependency. The second, equivalent 
recursive equation, shows how this would be computed from only the 
immediate successors of a transaction in $G$, without traversing 
the whole graph.
Note that $\pi(T) < c(T)$, and the values of $c(T)$ and $\pi(T)$ 
are fixed once $T$ has committed; $\pi$ will not change because 
committed $T$ only acquires new successors via forward edges, which do 
not influence $\pi(T)$.

The essence of SSN is a certification that prevents a transaction $T$ 
from committing if an exclusion window check fails for some direct
predecessor $U$:
\begin{definition}\label{def:exclusion-window}
  A dependency edge $U \xdep T$ in $G$ (or alternatively, transaction $U$) 
  violates the exclusion window of 
  $T$ if $\pi(T) \leq c(U) < c(T)$.
\end{definition}

The inequality checks whether $U$ (a predecessor of $T$ 
which committed first) might also be a successor of $T$ 
(because $U$ did not commit earlier than $T$'s oldest successor), 
indicating a potential cycle in $G$. When implementing exclusion window
checks, we can use two observations to simplify the process. First, we 
need only consider predecessors that committed before $T$ (the second 
inequality), which means the check can be completed during pre-commit of 
$T$ (regardless of what happens later). Second, of those predecessors that 
committed before $T$, we only need to examine the most recently-committed 
one. Using the following definition of $\eta(T)$, an exclusion window 
violation occurs if $\pi(T) \leq \eta(T)$, so $T$ must abort:
$$\eta(T) = max\left(\left\{c(U) : U \fdep T\right\} \cup 
\left\{-\infty\right\}\right)$$
We next illustrate visually why tracking $\pi(T)$ and enforcing exclusion 
windows might prevent cycles in $G$. Formal descriptions are provided 
later in Sect.~\ref{sec:correctness}.

Fig.~\ref{fig:ssn-in-pics}(a) gives a {\em serial-temporal} representation 
of a cycle in $G$. The horizontal axis gives the relative serial dependency order 
(as implied by the edges in $G$); the vertical axis gives the global commit order. 
In this figure, forward edges have positive slope (e.g., $T5 \xdep T1$), 
while back edges have negative slope (e.g., $T4 \xdep T3$). A transaction 
might appear more than once (connected by dashed lines, e.g., $T1$), 
if a cycle precludes a total ordering.

Visually, it is clear that $T1$ violates $T2$'s exclusion window because 
$\pi(T2) = c(T5) < c(T1) = \eta(T2)$. Fig.~\ref{fig:ssn-in-pics}(b) depicts 
information that is available to $T2$ as local knowledge. Without knowing 
$T1$'s predecessors, $T2$ must assume that $T1$ might also be a successor. 
Fig.~\ref{fig:ssn-in-pics}(c) demonstrates a case where the exclusion window 
is satisfied: $T3$ committed before $\pi(Tx)$---even earlier than $Tx$'s 
oldest successor---so $T3$ could not be a successor and $Tx$ will not close 
a cycle if committed; $T1$ cannot have any predecessor newer than $\pi(Tx)$ 
as that would violate its own exclusion window; any later transactions that 
links $T1$ with $Tx$ would suffer an exclusion window violation.

Finally, Fig.~\ref{fig:ssn-in-pics}(d) illustrates a false positive case, 
where $T3$ aborts due to an exclusion window violation, even though no cycle 
exists. We note, however, that allowing $T3$ to commit would be dangerous: 
some predecessor to $T1$ might yet commit with a dependency on $T4$, 
closing a cycle without triggering any additional exclusion window violations.

\subsection{Safe retry}
Users submit transactions supposing they will commit, however, the underlying 
CC scheme might abort transactions due to various reasons, such as write-write 
conflicts. Ideally, the CC scheme should ensure that all transactions eventually
commit (perhaps after some number of automatic retries), 
unless the user requests an abort.

SSN exhibits the safe retry property~\cite{pg12}. 
Suppose SSN aborts transaction $T$ 
because $U$ violates its exclusion window, and that the user retries 
immediately with $T'$. 
Any back-edge successor $T$ had, is a transaction $S$ that committed before $T$.
Since $T'$ started after $S$ committed, $T'$ will not read data that will be
overwritten by $S$. That is, $T'$ will not have the same successor, and the same
set of dependencies cannot form for $T'$.

The importance of safe retry is often overlooked, 
and many serializable schemes do {\em not} provide this property, 
including 2PL ($T'$ could deadlock with the winner of a previous deadlock) 
and OCC~\cite{hekaton,silo} that relies on read set validation (the 
overwriter could still be in progress, causing another failure).
In Sect.~\ref{sec:eval}, we empirically evaluate this property for SSN and
other CC schemes.

\begin{figure*}[t]
\includegraphics[width=\textwidth]{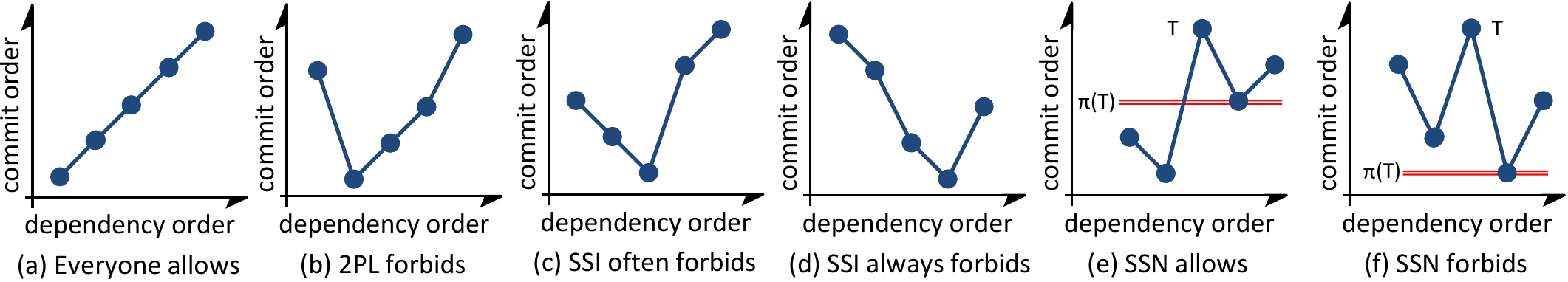}
\caption{SSN allows all schedules (a--e) that do not have ``peaks,'' and 
also ``peaks'' where no predecessor of $T$ violates the exclusion window. 
SSN rejects only case (f); other schemes tend to reject the ``valleys'' that arise frequently under 
MVCC.}
\label{fig:ssn-scenarios}
\end{figure*}

\subsection{Correctness}
\label{sec:correctness}
We now formally prove the correctness of SSN. Based on the database model 
we set up in Section~\ref{sec:background}, 
we first recall the key result of serialization theory:

\begin{theorem}
\label{acyclic-thm}
Let an execution with schedule $h$ have a serial dependency graph $G(h)$ 
with no cycles. Then the execution is serializable\footnote{There are many 
formulations such as \cite{BernsteinSW79} and \cite{Papadimitriou79b}, 
the presentation with this form of dependency definition is in \cite{adya99}.}.
\end{theorem}

As mentioned in previous sections, SSN requires the underlying CC scheme 
forbid lost writes and dirty reads:

\begin{definition}\label{certifiable-cc}
Let a \textit{certifiable} scheduler be any CC scheme that forbids lost 
writes and dirty reads (other than a transaction reading its own writes).
\end{definition}

Definition~\ref{certifiable-cc} effectively allows any CC scheme at least 
as strong as RC. In particular, the underlying CC scheme is free to return 
any committed version from a read (not necessarily in a repeatable fashion), 
and can delay accesses arbitrarily.

Given a non-serializable schedule $h$ produced by a certifiable scheduler, 
we first identify the ``dangerous'' edges in its dependency graph $G(h)$ 
that SSN targets. We then prove that these edges exist in any dependency 
cycle that arises under a certifiable scheduler. We argue the correctness 
of SSN as follows:
\begin{theorem}
\label{exclusion-edge}
Let $h$ be any non-serializable history produced by a certifiable scheduler. 
Then the dependency graph $G(h)$ contains at least one exclusion window violation.
\end{theorem}

\begin{proof}
By the hypothesis that $h$ is non-serializable and Theorem~\ref{acyclic-thm}, 
$G(h)$ must contain a cycle involving $n \geq 2$ transactions.\footnote{We 
ignore self loops, since our model excludes them. In reality transactions 
should be allowed to read their own writes.} We first name the transactions 
in that cycle, so that $T_n$ committed first in time: $T_n\nobreak \xdep\nobreak 
T_1\nobreak \xdep\nobreak T_2\nobreak \xdep\nobreak \ldots\nobreak 
\xdep\nobreak T_{n-1}\nobreak \xdep\nobreak T_n$. Because $T_n$ committed 
first in the cycle, its predecessor---which is also a successor---must be 
reached by a back edge. We can choose the lowest value of $k$ such that 
$T_k\nobreak \rdepstar \nobreak T_n$ holds. Then $\pi(T_k) \leq c(T_n)$. 
Further, the predecessor of $T_k$ ($T_{k-1}$, or $T_n$ if $k=1$) must be 
reached by a forward edge. Combining the two facts reveals an exclusion 
window violation: $\pi(T_k)\nobreak \leq\nobreak c(T_n)\nobreak \leq\nobreak 
c(T_{k-1})\nobreak <\nobreak c(T_k)$. Since we have shown that $T_k$ always 
exists and always has a predecessor that violates $T_k$'s exclusion window, 
we conclude that $G(h)$ always contains an exclusion window violation.\qed
\end{proof}

\begin{definition}
\label{thm:ssn-works}
A certifiable scheduler is said to \textit{apply SSN certification} if it 
aborts any transaction $T$ that, by committing, would introduce an 
exclusion window violation into the dependency graph. That is, 
SSN forces $T$ to abort if there exists a potential edge $U \xdep T$ where 
$\pi(T) \leq c(U) < c(T)$.
\end{definition}

\begin{theorem}
All executions produced by a certifiable scheduler that applies SSN 
certification are serializable.
\end{theorem}

\begin{proof}
By contradiction: If there is any execution of the scheduler that is 
non-serializable, Theorem \ref{exclusion-edge} shows that there is an 
edge in the dependency graph that violates the exclusion window.
However, the certification check in the scheduler does not allow any 
such edge to be introduced.\qed
\end{proof}

We next formally prove SSN's safe retry property. Suppose SSN aborts
transaction $T$ because $U$ violates its exclusion window, and that the user
retries immediately with $T'$ . Then the same dependencies
cannot force $T$ to abort (though other newly arrived transactions could
produce a new exclusion window violation for $T'$).

\begin{theorem}
SSN provides the ``safe retry'' property, assuming the underlying CC scheme 
does not allow $T$ to see versions that were overwritten before $T$ began.
\end{theorem}

\begin{proof}
An exclusion window violation requires $U \xdep T \rdep S$, 
where $\pi(T) < c(U)$ and $T$ read a value which $S$ will eventually overwrite.
By the definition of back edge, $S$ committed before $T$ tried to commit,
therefore, before $T'$ starts. Therefore, $T'$ will read the version $S$ created
(or a later one), so no anti-dependency between $T'$ and $S$ will be created.
That is, the situation that caused $T$ to have an exclusion window violation
will not occur for $T'$, although other dependencies might form
for $T$, and another violation may occur.\qed
\end{proof}

\subsection{Discussion}
\label{app-discuss}
We now compare SSN with other cycle prevention schemes, and reason about 
their relative merits. Fig.~\ref{fig:ssn-scenarios} highlights several 
``shapes'' that transaction dependencies can take when plotted in the 
serial-temporal form. Of all the serializable schedules shown, SSN rejects 
only the last. In contrast, 2PL admits only the first (all others contain 
forbidden back edges). SSI always admits cases (a) and (b), always rejects 
(d), and often rejects (c) and (f).\footnote{SSI allows (c) if the leftmost 
transaction is read-only and sufficiently old, but rejects (f) if a 
(harmless) forward anti-dependency edge joins $T$ with its predecessor.} 
Case (e) cannot even arise under SI, let alone SSI. Thus, the improved cycle 
test in SSN allows it to tolerate a more diverse set of transaction profiles 
than existing schemes, including schedules forbidden by SI.

Fig.~\ref{fig:nasty-swimlane} illustrates the accuracy of SSN in a different way 
using a simple schedule involving only three transactions, 
with time flowing downward (a). On the surface, 
all might appear reasonable: each transaction makes some number of reads 
before writing one record and committing. However, once the execution passes the 
horizontal dotted line, many serializable CC strategies are doomed to abort at least 
one transaction.
\begin{figure}[t]
\includegraphics[width=\columnwidth]{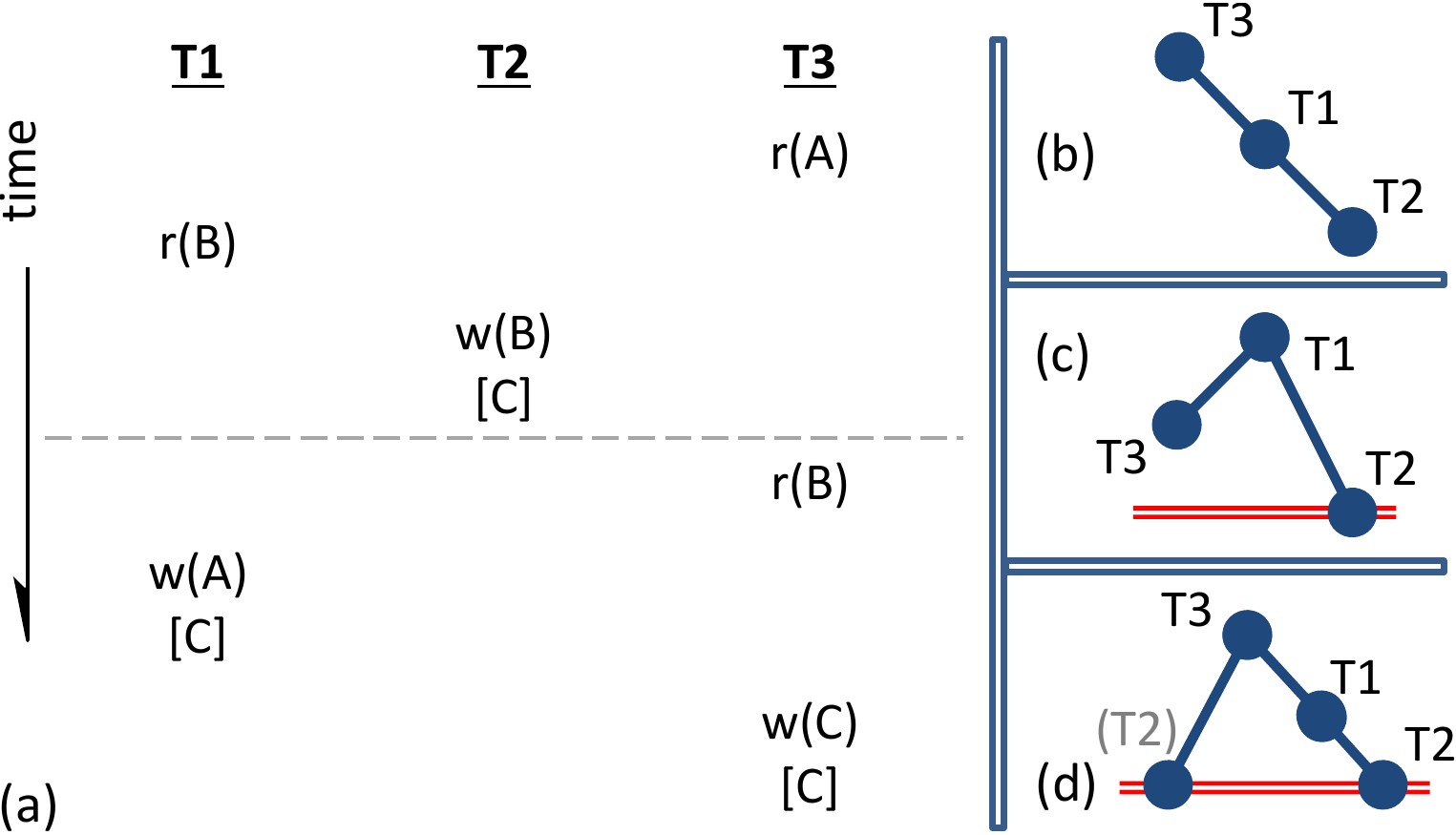}
\caption{A pathological scenario that will deadlock under 2PL but be serializable under SI (a); 
the resulting serial-temporal plot when T3 commits last under SI (b) 
and when T1 commits last (c); the serial-temporal plot that results under RC (d).
The horizontal and vertical axes of (b)--(d) represent dependency and commit orders, respectively.}
\label{fig:nasty-swimlane}
\end{figure}

2PL will deadlock: $T1$ reads B, blocking $T2$ which in turn blocks $T3$.
Meanwhile, $T1$ blocks attempting to write A, which $T3$ has read-locked.
With SI-based schemes, $T3$'s read of B will return the original version
rather than the one produced in $T2$, so there will be a back edge $T3 \rdep
T2$ (as well as $T1 \rdep T2$). Thus SI-based certifiers that check for single
back edges will abort at least one of the transactions, and SSI will also
abort one due to the dangerous structure $T3 \rwdep T1 \rwdep T2$ where $T2$
committed first and $T3$ is not read-only. In contrast, SI+SSN safely allows
all three transactions to commit, with the dependency structure shown in
Fig.~\ref{fig:nasty-swimlane}(b). SI+SSN is not a perfect certifier, as it sometimes aborts
transactions unnecessarily: if $T1$ tries to commit last (after $T3$ commits),
then SI+SSN would abort $T1$ with a failed exclusion window test because
$\pi(T1) \leq c(T2) < c(T3) < c(T1)$, even though the schedule is actually
serializable. Fig.~\ref{fig:nasty-swimlane}(c) depicts the dependency structure for this case.
Finally, suppose the concurrency control is RC+SSN and the schedule is
that of Fig.~\ref{fig:nasty-swimlane}(a). Now, $T3$'s read of B sees the version
produced in $T2$'s write, and so the dependencies have a cycle, with 
$T2 \fdep T3$ as well as $T3 \rdep T1 \rdep T2$. 
Thus $\pi(T3) = c(T2)$; and because $T2$ is also a potential successor of $T3$,
$T2$ violates the exclusion window test for $T3$, and this will force $T3$ to
abort, thereby preventing the non-serializable execution.
Fig.~\ref{fig:nasty-swimlane}(d) shows the dependencies for this case.

Despite this pathological case being simple, it 
gives an intuitive explanation for why SSN works 
well compared to other schemes: most schemes identify and reject the 
existence of back edges (either singly or in pairs) as a necessary 
condition to close cycles. However, we have seen that the ``peaked'' 
deadly structure identified earlier is a more precise cycle-closing 
pattern that allows SSN to ignore a large fraction of harmless back edges while 
still detecting all harmful ones.

As a final observation, we expect write-intensive workloads 
to perform better under RC+SSN than
under SSI: A major source of transaction failures under SSI is temporal 
skew, where a transaction attempts to overwrite a version created after 
its snapshot. By allowing transactions to always access the latest version 
(except when forbidden by SSN), RC should lower the risk of encountering 
temporal skew in a short transaction. 
We show this effect in
Sect.~\ref{subsec:simulation-rw-ratio}.

\section{SSN protocols in multi-version systems}
\label{sec:implementation}
In this section, we describe how SSN can be implemented for multi-version
systems, including disk-based and main-memory optimized ones. Specifically, we
discuss how SSN processes each read, write and commit request. We assume that
each version and transaction is associated with some storage to store the
metadata SSN requires. To overlay SSN on top of a single-version CC scheme (e.g.,
Read Committed with Locking), one will need to store information in lock entries
as proxies for the versions and keep some locks (in non-blocking modes) longer
than the underlying CC would have done. This is similar to PostgreSQL's SSI
implementation~\cite{pg12}. We leave lock-based SSN as future work, and in this
paper we focus on multi-version systems.

The rest of this section first gives the basic SSN protocols that can be made
parallel using latches. We then describe how the protocols can be parallelized in
a latch-free manner that is suitable for modern main-memory databases.

\subsection{Basic protocols}
\label{subsec:basic-prot}
The basic protocols of SSN require space and computation linearly proportional
to the combined read/write footprints of all in-flight transactions, 
plus constant space per version. 
Each transaction should maintain its footprints using read and write sets, 
which contain all the versions read and written by the transaction, respectively. 
SSN summarizes dependencies between transactions using various timestamps that 
correspond to commit times. For in-flight and recently-committed transactions, 
these timestamps can be stored in the transaction's context. 
For older transactions, the timestamps can be maintained in versions without a 
need to remember which committed transactions were involved. 
SSN supports early detection of exclusion window violations before entering
pre-commit, aborting the transaction immediately if the arrival of a too-new
(too-old) potential predecessor (successor) dooms it to failure.

Suppose transaction $T$ created version $V$, while transactions $R$ and $W$ respectively 
read and overwrote $V$.
Then we can define $c(V) = c(T)$, $\pi(V) = \pi(W)$, and 
$\eta(V) = max\left(\left\{c(R) : T \fdep R \right\} \cup \left\{c(T)\right\}\right)$.
These per-version timestamps maintained in each version record everything we need to implement SSN if transaction 
execution is single-threaded. Explicit dependency tracking (making transactions aware of 
other transactions) is only needed to avoid races between transactions that coexist in time, 
particularly those whose validation phases overlap.

Table~\ref{tab:ssn-metadata} summarizes the metadata 
(along with the corresponding initial values) 
which SSN tracks for each transaction $T$ and version $V$. 
Version-related states persist for the life of the version, while 
transaction states are discarded soon after the transaction ends. 
Although SSN increases per-version space overhead, we note that many MVCC implementations 
already track some of these values.\footnote{For example, PostgreSQL maintains the 
equivalent of \texttt{v.cstamp} and \texttt{v.prev}. 
In each version, SSN takes an extra of 16 bytes for \texttt{sstamp} and \texttt{pstamp}, 
assuming 8-byte stamps.
For one million versions, 
SSN needs in total less than 16MB of extra memory.
This is likely tolerable in today's systems with abundant memory and storage.}

\begin{table}[t]
\centering
\setlength{\tabcolsep}{3pt}
\begin{tabular}{l l c}
\hline
{\bf Value} & {\bf Meaning} & {\bf Initial value} \\\hline
\texttt{t.cstamp} & Transaction end time, $c(T)$ & 0 \\
\texttt{t.status} & In-flight/committed/aborted & in-flight\\
\texttt{t.pstamp} & Predecessor high watermark, $\eta(T)$ & 0 \\
\texttt{t.sstamp} & Successor low watermark, $\pi(T)$ & $\infty$ \\
\texttt{t.reads} & Non-overwritten read set & $\emptyset$ \\
\texttt{t.writes} & Write set & $\emptyset$ \\\hline
\texttt{v.cstamp} & Version creation stamp, $c(V)$ & ``invalid'' (0)\\
\texttt{v.pstamp} & Version access stamp, $\eta(V)$ & 0 \\
\texttt{v.sstamp} & Version successor stamp, $\pi(V)$ & $\infty$\\
\texttt{v.prev} & Pointer to the overwritten version & \texttt{NULL} \\\hline
\end{tabular}
\caption{Metadata required by SSN.}
\label{tab:ssn-metadata}
\end{table}

\subsubsection{Interactions with the underlying CC}
As we have discussed previously, 
it is the underlying CC that dictates which version a transaction should see. 
Therefore, SSN's read and write protocols (\texttt{ssn\_read} and \texttt{ssn\_write} 
functions in Algorithm~\ref{alg:ssn-rw}, respectively) 
receive a reference to the version returned by the underlying CC 
as a parameter.
It is up to the underlying CC to employ the appropriate synchronization mechanism 
to guarantee correct interactions among threads. 
For example, in an SI implementation the worker thread could walk through the
version chain to find the latest committed version that is visible to the transaction, 
and then pass the desired version to \texttt{ssn\_read}. 
The SI implementation could indicate that a version is not yet committed by 
storing the creator transaction ID (TID) in the version's commit timestamp field. 
Readers who see a version with a TID in the commit stamp field will skip and 
continue to examine the next available version. 
Upon commit, the creator transaction will transform the TID to the real commit timestamp. 
Some recent systems~\cite{hekaton,ERMIA} follow this paradigm.
As a result, 
\texttt{ssn\_read} itself needs no extra synchronization protocol. 
It always reads a version that is already committed and made immutable by the creator transaction. 

Neither does \texttt{ssn\_write} need to handle concurrent writes itself: 
the underlying CC determines whether a new version can be appended, 
possibly by latching the version chain and compare transaction/version timestamps. 
If the transaction successfully installs a new version, 
as part of the underlying CC's write protocol, 
\texttt{v.prev} should point to the version that is overwritten. 
The transaction then proceeds to update the timestamps using \texttt{ssn\_write}.
The underlying CC ensures that the in-flight new version is invisible to concurrent reader transactions, 
e.g., by storing the creator's TID in the version's commit timestamp field as described earlier.

Different from the read and write protocols, 
SSN's commit protocol needs proper synchronization among transactions with overlapping footprints. 
We discuss the details following SSN's read and write protocols. 

\begin{algorithm}[t]
\begin{lstlisting}[language=python,morekeywords={invalid,valid}]
def ssn_read(t, v):
  if v not in t.writes:
    # update \eta(t) with w:r edges
    t.pstamp = max(t.pstamp, v.cstamp)

    if v.sstamp is infinity:
      t.reads.add(v)  # no overwrite yet
    else:
      # update \pi(t) with r:w edge
      t.sstamp = min(t.sstamp, v.sstamp)
    verify_exclusion_or_abort(t)

def ssn_write(t, v):
  if v not in t.writes:
    # update \eta(t) with w:r edge
    t.pstamp = max(t.pstamp, v.prev.pstamp)
    t.writes.add(v)
    t.reads.discard(v)  # avoid false positive
    verify_exclusion_or_abort(t)
\end{lstlisting}
\caption{SSN read and write protocols (for multi-version systems).}
\label{alg:ssn-rw}
\end{algorithm}

\subsubsection{Read}
Lines 1--11 of Algorithm~\ref{alg:ssn-rw} describe SSN's read protocol.
Besides the reading transaction \texttt{t}, 
it also receives a reference to the appropriate version returned by the underlying CC as a parameter. 
Transaction $T$ will record in \texttt{t.pstamp} the 
largest \texttt{v.cstamp} it has seen to reflect $T$'s dependency on the 
version's creator (line 4). To record the read anti-dependency from the transaction that 
overwrote $V$ (if any), $T$ records the smallest \texttt{v.sstamp} in 
\texttt{t.sstamp} (lines 9--10). 
As shown by line 7 of Algorithm~\ref{alg:ssn-rw}, if the version has not yet been 
overwritten, it will be added to $T$'s read set and checked for late-arriving 
overwrites during pre-commit. 
The transaction then verifies the exclusion window and
aborts if a violation is detected. 
The transaction will transition from in-flight status to the aborted status.
If the transaction is aborted, 
the safe retry property allows it to retry immediately, minimizing both wasted 
work and latency.

Note that $T$ does not track reads of versions it creates or overwrites, 
nor does it track reads if an overwrite has already committed 
(i.e., \texttt{v.sstamp} is valid). 
The read and write sets are currently implemented as simple arrays/vectors in our prototype. 
Further, the read set does not need to be searchable in order to enforce 
repeatable reads: SSN automatically enforces repeatable reads because a non-repeatable 
read corresponds to the cycle $T \rwdep W \wrdep T$. While a practical implementation would 
be well-advised to enforce repeatable reads by less draconian means, 
this is a matter of performance optimization, not correctness.

\subsubsection{Write}
\label{subsubsec:ssn-write}
The write protocol is shown by lines 13--19 of Algorithm~\ref{alg:ssn-rw}.
Note that \texttt{v} in 
\texttt{ssn\_write} refers to the new version generated by $T$. 
When updating a version, $T$ updates its predecessor timestamp {\tt t.pstamp} with 
{\tt v.prev.pstamp}. 
We use {\tt v.prev.pstamp} rather than {\tt v.prev.cstamp} because a write will never cause 
inbound read anti-dependencies, but it can trigger outbound read anti-dependencies 
(i.e., $R \rwdep T$, in which $R$ read $V$ before $T$ overwrote it). 
$T$ then records $V$ in its write set for the final validation at pre-commit (line 17).
If more reads came later, $T$ would update {\tt t.pstamp} with {\tt v.prev.pstamp}, 
which were updated by readers that came after $T$ but installed the new version {\tt v} 
before $T$ entered pre-commit.
Additionally, we must remove $V$ from $T$'s read set, if present: updating $\pi(T)$ using 
the edge $T \rwdep{} T$ would violate $T$'s own exclusion window and 
trigger an unnecessary abort. Sect.~\ref{subsubsec:finalize-pi} describes how we 
efficiently ``remove'' $V$ from $T$'s read set by skipping
processing $V$ when examining the read set, without having to make the read set
searchable.

\subsubsection{Commit}
\label{subsubsec:ssn-commit}
\begin{algorithm}[t]
\begin{lstlisting}[language=python]
def ssn_commit(t):
  t.cstamp = next_timestamp()  # begin pre-commit

  # finalize \pi(T) 
  t.sstamp = min(t.sstamp, t.cstamp)
  for v in t.reads:
    t.sstamp = min(t.sstamp, v.sstamp)

  # finalize \eta(T) 
  for v in t.writes:
    t.pstamp = max(t.pstamp, v.prev.pstamp)

  verify_exclusion_or_abort(t)
  t.status = COMMITTED  # post-commit begins

  for v in t.reads:  # update \eta(V)
    v.pstamp = max(v.pstamp, t.cstamp)
  for v in t.writes:
    v.prev.sstamp = t.sstamp # update \pi(V)
    # initialize new version
    v.cstamp = v.pstamp = t.cstamp
\end{lstlisting}
\caption{SSN commit protocol (for multi-version systems).}
\label{alg:ssn-commit}
\end{algorithm}
We divide the commit process into two phases: pre-commit and post-commit.
During pre-commit we first finalize $\pi(T)$ and $\eta(T)$, and then test the exclusion window. 
If the exclusion window is not violated, 
$T$ commits and the system propagates appropriate timestamps 
into affected versions during the post-commit phase. 
Pre-commit begins when $T$ requests a commit timestamp $c(T)$ in the in-flight status 
(set at transaction initialization), 
which determines its global commit order, as depicted by line 2 of Algorithm~\ref{alg:ssn-commit}. 
After initializing $c(T)$, $T$ is no longer allowed to perform reads or writes. 
It then computes $\pi(T)$, following the formula given in Section~\ref{sec:ssn}.
The computation only considers $\pi(V)$ of reads that were overwritten before $c(T)$.

The transaction next computes $\eta(T$) using a similar strategy, 
but must account for more dependency edge types. 
Recall that $T$ can acquire predecessors in two ways: 
reading or overwriting a version causes a dependency on the transaction that created it; 
overwriting a version also causes a dependency on all readers of the overwritten version. 
The read and write protocols account for the former by checking $c(V)$, 
and pre-commit accounts for the latter using $\eta(V)$.

Once $\pi(T)$ and $\eta(T)$ are both available, a simple check for $\pi(T) \leq \eta(T)$ 
identifies exclusion window violations. 
As shown by line 14 of Algorithm~\ref{alg:ssn-commit}, 
transactions having $\eta(T) < \pi(T)$ are allowed to commit, transitioning to the committed status. 
Otherwise, the transaction would abort with an ``aborted'' status and remove any new versions it installed during forward processing. 
The \texttt{pstamp} and \texttt{sstamp} maintained during forward processing 
and pre-commit are discarded, 
as if the transaction was never processed.
During post-commit, the transaction updates $c(V)$ for each version it created, 
$\pi(V)$ for each version it overwrote, and $\eta(V)$ for each non-overwritten version it read.
Post-commit is a clean-up and resource reclamation operation that does not cause extra 
transaction aborts. 
Versions created by a pre-committed transaction that is still in post-commit do not delay transactions. 
A concurrent transaction can access the versions and infer their timestamps by inspecting 
the corresponding pre-committed transactions' metadata.

The commit protocol described above allows parallel transaction execution but itself executes serially. 
The caller should hold a latch upon entering \texttt{ssn\_commit} and release the latch after finished 
executing the function.
The restriction is due to a number of races that can arise between acquiring $c(T)$, 
computing $\pi$ and $\eta$ for the transaction, 
and updating \texttt{sstamp} and \texttt{pstamp} for versions in the transaction's read and write sets.

\subsection{Latch-free parallel commit}
\label{subsec:parallel-commit}
Latch-based serial validation imposes an unacceptable scalability penalty,
as shown by recent research~\cite{jhf+13}, 
and especially so for modern main-memory optimized
systems~\cite{silo,FOEDUS,ERMIA,Deuteronomy,hekaton}. 
In the rest of this section, 
we describe how SSN's commit protocol can be parallelized in a latch-free manner for 
recent main-memory systems. 
Although not as significant as in main-memory systems, 
conventional disk-based systems can also benefit from our approach, 
with appropriate adjustment to a few assumptions described later. 
In this paper, we focus on main-memory systems.

\subsubsection{Main-memory OLTP systems}
\label{subsubsec:mmdbms}
The abundant amount of memory available in modern severs have lead to many recent 
main-memory OLTP systems~\cite{ERMIA,silo,FOEDUS,hekaton,Deuteronomy}. 
These systems assume that at least the working set (if not the whole database) resides in memory, 
thus allowing several important optimizations. 
First, a thread can execute a transaction from beginning to the end without any context switch. 
Heavyweight I/O operations are completely out of the critical path. 
In case the transaction needs to fetch data from storage, 
mechanisms such as Anti-Caching~\cite{AntiCaching} will abort and restart the transaction 
when all the data needed is available in memory. 
Second, main-memory OLTP systems utilize massively parallel
processors and large main memory more effectively, 
by using memory-friendly data structures and algorithms that utilize high
parallelism.
For examples, most main-memory systems dispense with centralized locking 
and co-locate the locks with records~\cite{VLL,FOEDUS}. 
Lock-free techniques are often used to obtain high CPU
utilization~\cite{BwTree,ERMIA}. 

We exploit atomic primitives provided by the hardware to devise our latch-free
parallel commit protocol.
In particular, we assume 8-byte atomic reads/writes and the availability of 
the {\tt compare\_and\_swap} (\texttt{CAS}) instruction;
both are supported by most modern parallel processors.
We also assume that during transaction execution, there will be no I/O
operations on the critical path and a thread will not be re-purposed.

\subsubsection{Finalizing $\pi$}
\label{subsubsec:finalize-pi}
As shown in Algorithm~\ref{alg:ssn-commit}, a committing transaction $T$ needs to calculate
$\pi(T)$ and $\eta(T)$ in a critical section (lines 4--11).
Recall that $T$ needs to update {\tt t.sstamp} with {\tt v.sstamp}, 
which is set by the transaction that overwrote $V$.
In the latch-based commit protocol (Algorithm~\ref{alg:ssn-commit}), 
only one transaction can examine {\tt v.sstamp} at the same time, 
and no concurrent write to {\tt v.sstamp} is allowed.
A latch-free protocol, however, must account for transactions that are concurrently
committing: {\tt v.sstamp} might be changed at any time by the overwriter $U$,
which might have acquired a commit timestamp earlier or later than $T$'s.
In case $U$ acquired a commit timestamp that is earlier than {\tt t.cstamp}, $T$ should
update its {\tt t.sstamp} with $U$'s successor low watermark (i.e., {\tt
u.sstamp}) if \texttt{u.sstamp} is smaller than \texttt{t.sstamp}.

As we have discussed in Section~\ref{subsec:basic-prot}, 
an SI implementation can indicate that a version is not yet available for reading by 
storing the owner transaction's TID on the version's \texttt{cstamp} field.
This approach is applicable to solving the above problem as well:
an updater $U$ should install its TID in the overwritten version 
$V$'s {\tt sstamp} after it successfully
installs a new version (i.e., as of the write before entering pre-commit).
Our implementation does this in \texttt{ssn\_write}. 
$U$ then changes {\tt v.sstamp} to contain {\tt u.sstamp} during $U$'s post-commit phase.
Consequently, a concurrent transaction might read an \texttt{sstamp} field and see 
a TID when reading a version or committing.\footnote{
Storing TID (before the overwriter finalizes) in the overwritten version's \texttt{sstamp} also eases 
the removal of updated versions from the read set 
(see the end of Section~\ref{subsubsec:ssn-write}). 
When iterating the read set, 
the updater $T$ simply skips versions whose 
\texttt{sstamp} points to $T$'s own TID.} 
For the former case, 
the transaction simply treats it in the same way as if \texttt{v.sstamp} contained~$\infty$ 
in \texttt{ssn\_read} 
(i.e., no overwrite yet) and adds the tuple to its read set. 
Lines 9--18 of Algorithm~\ref{alg:ssn-parallel-commit} show how the latter case is handled in 
our latch-free commit protocol. 
After $T$ detects a TID in {\tt v.sstamp} (line 9), 
it will obtain $U$'s transaction context (line 10) through a centralized transaction table indexed by TIDs.
Note that since \texttt{v.sstamp} might be changed to contain the overwriter's commit timestamp from its TID 
at any time, 
at line 8 we must read \texttt{v.sstamp} into \texttt{v\_sstamp} and determine if we should proceed 
to line 10 using \texttt{v\_sstamp}. 
Recall that by Definition~\ref{def:exclusion-window}, $\pi(T) < c(T)$. 
So if $U$ has acquired a commit timestamp earlier than {\tt t.cstamp}, $T$ has to wait for $U$ to 
conclude so that $U$'s successor low-watermark is stable.\footnote{There exists other
viable approaches other than spinning for $U$'s \texttt{sstamp} to become stable;
e.g., $T$ could deduce $U$'s low-watermark by helping with $U$'s 
pre-commit phase by iterating over $U$'s read set.
In experiments, we use spinning due to its simplicity and lightweightness.}
As line 12 shows, if $U$ has not entered pre-commit, then it will have a commit timestamp larger than $T$'s. 
Otherwise, we spin at line 13 until \texttt{u.cstamp} contains $U$'s commit timestamp in \texttt{u.cstamp}. 
Note that such spinning is necessary, 
because \textit{acquiring} a commit timestamp (line 3) and 
\textit{storing} it in \texttt{u.cstamp} are not done atomically by a single instruction. 
For example, \texttt{next\_timestamp} itself might draw a timestamp form a centralized counter using an 
atomic \texttt{fetch-and-add} instruction. 
Therefore, transactions entering pre-commit must first transition to a ``committing'' status, 
and then obtain a commit timestamp (lines 2--3). 
As a result, upon detecting $U$ has entered pre-commit (i.e., \texttt{status} is not \texttt{INFLIGHT}), 
\texttt{u.cstamp} is guaranteed to contain a valid commit timestamp eventually. 
Note that this ``committing'' status is not needed in the latch-based commit protocol 
(Section~\ref{subsubsec:ssn-commit}): 
only a single thread can execute \texttt{ssn\_commit} at a time; 
a transaction's status is not exposed to other concurrent transactions.

\begin{algorithm}[t]
\begin{lstlisting}[language=python,morekeywords={spin_while}]
def ssn_parallel_commit(t):
  t.status = COMMITTING
  t.cstamp = next_timestamp()  # begin pre-commit

  # finalize \pi(T) 
  t.sstamp = min(t.sstamp, t.cstamp)
  for v in t.reads:
    v_sstamp = v.sstamp
    if is_TID(v_sstamp):
      u = get_transaction(v_sstamp)
      # obtain u.cstamp
      if u.status is not INFLIGHT:
        spin_while(u.cstamp == 0)
        if u.cstamp < t.cstamp:
          # wait for U to finish pre-commit
          spin_while(u.status == COMMITTING)
          if u.status == COMMITTED:
            t.sstamp = min(t.sstamp, u.sstamp)
    else:
      t.sstamp = min(t.sstamp, v.sstamp)

  # finalize \eta(T) 
  for v in t.writes:
    for r in v.prev.readers
      if r.status is not INFLIGHT:
        spin_while(r.cstamp == 0)
        if r.cstamp < t.cstamp:
          spin_while(r.status == INFLIGHT)
          if r.status = COMMITTED:
            t.pstamp = max(t.pstamp, r.cstamp)
    # re-read pstamp in case we missed any reader
    t.pstamp = max(t.pstamp, v.prev.pstamp)

  verify_exclusion_or_abort(t)
  t.status = COMMITTED  # post-commit begins

  for v in t.reads:  # update \eta(V)
    pstamp = v.pstamp
    while pstamp < c.cstamp
      if CAS(v.pstamp, pstamp, t.cstamp):
        break
      pstamp = v.pstamp

  for v in t.writes:
    v.prev.sstamp = t.sstamp # update \pi(V)
    # initialize new version
    v.cstamp = v.pstamp = t.cstamp
\end{lstlisting}
\caption{Latch-free SSN commit protocol (for multi-version systems).}
\label{alg:ssn-parallel-commit}
\end{algorithm}

After obtaining \texttt{u.cstamp}, 
the protocol continues to check if \texttt{u.cstamp} is smaller than \texttt{t.cstamp}.
If so, $T$ needs to find out $U$'s final result (line 16): 
if $U$ indeed pre-committed, $T$ then updates {\tt t.sstamp} with {\tt u.sstamp} to finalize its
successor low-watermark (lines 17--18). 
If, however, $U$ acquired a commit timestamp later than {\tt t.cstamp}
or has not even started its pre-commit phase, 
$T$ can continue to process the next element in the read set without spinning on $U$ at line 16.

\subsubsection{Finalizing $\eta$}
We first note a fundamental difference between {\tt v.sstamp} and {\tt v.pstamp} that determines
how their calculations can be parallelized: during the lifetime of $T$,
$V$ could only have at most one successful overwriter, i.e., the transaction that 
installed a new version of $V$ and committed.
Before $T$ enters pre-commit, however, multiple concurrent reader transactions 
(denoted as ``readers'' for simplicity) could have read $V$ during $T$'s 
lifetime. Each of these concurrent readers will need to update {\tt v.pstamp} as shown by lines 16--17 of
Algorithm~\ref{alg:ssn-commit} (the latch-free version is described later).
In a latch-free commit protocol, as a result, $T$ will have to take into consideration all possible
concurrent readers that have obtained a commit timestamp earlier than {\tt t.cstamp}.
Essentially {\tt v.pstamp} becomes an array of the commit timestamps of all the readers of $V$.

A simple implementation might convert {\tt v.pstamp} directly to an array of timestamps, one per
predecessor, bloating the amount of metadata needed per version. To make the readers-tracking efficient,
we use a bitmap to summarize all of $V$'s readers. Each bit corresponds to one reader transaction.
As we have discussed in Section~\ref{subsubsec:mmdbms}, 
in most main-memory databases, transactions are rarely delayed (e.g., by I/O operations).
A worker thread processes only one transaction at a time, 
and there is roughly one worker thread per CPU core.
We use this fact and correspond the $i$-th bit in the bitmap to thread/transaction $i$.
Whenever a thread reads $V$ on behalf of a transaction $R$ for the first time, the thread registers itself
by setting the corresponding bit in {\tt v.readers}. R clears the bit after it concludes. 

Fig.~\ref{fig:tx-table} shows an example of three threads, each executing on behalf of a different transaction.
In this example, Thread~1 created tuple version \texttt{v1} and has already committed.
Thread~0 appended a new version, \texttt{v2}, after \texttt{v1} but has not yet entered pre-commit.
At the same time, Thread~2 read \texttt{v1} and registered itself in \texttt{v1.readers} by
setting the third least significant bit. As a result, when Thread~0 tries to commit, it will be able
to locate the transaction being executed by Thread~2 by following \texttt{v1.readers} and locate
its \texttt{cstamp} in the transaction-thread mapping table.

\begin{figure}[t]
\includegraphics[width=\columnwidth]{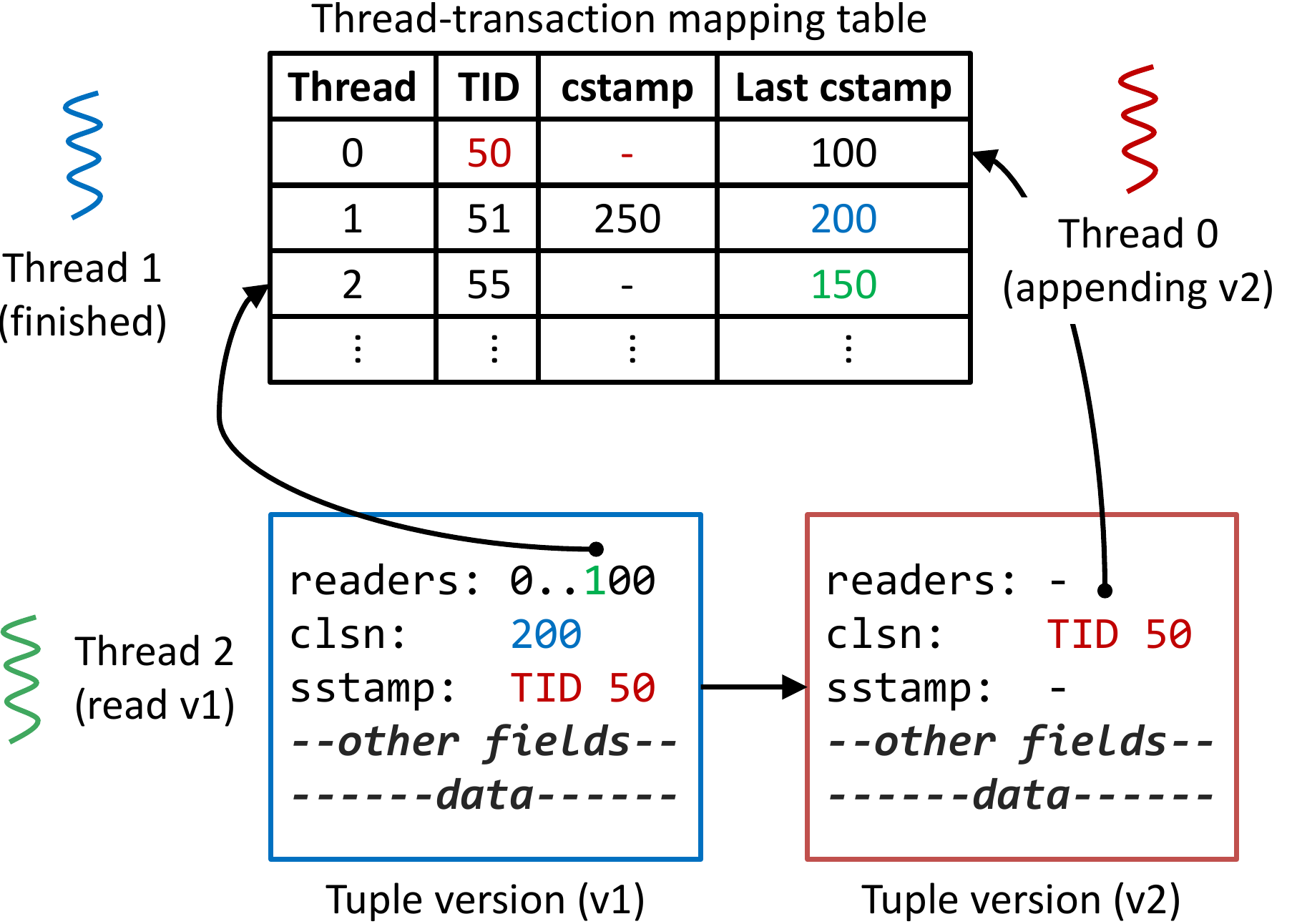}
\caption{The bit positions in the readers bitmap serve as indexes to the centralized transaction table
which records details on the transaction that is being run by each thread (we discuss the use of
the ``last cstamp'' field later in Sect.~\ref{subsec:read-opt}).}
\label{fig:tx-table}
\end{figure}

With a readers bitmap in each version, $T$ is able to examine all the concurrent readers of $V$
to finalize {\tt t.pstamp}. The details are described by lines 23--32 of Algorithm~\ref{alg:ssn-parallel-commit}.
For brevity, in the algorithm we omit the details of extracting the thread ID from the readers bitmap.
Our current implementation uses the bit scan reverse (\texttt{BSR}) instruction available on x86 platforms~\cite{IntelManual}.
As shown by lines 24--30, for each overwritten version, $T$ examines and waits
for each reader that acquired an earlier commit timestamp to finish pre-commit, 
using the same spinning machinery introduced earlier for finalizing $\pi$. 
If the reader successfully committed, $T$ will
update {\tt t.pstamp} using {\tt r.cstamp}.

If we accounted concurrent readers in an array, garbage collection is needed to remove unneeded readers metadata from $V$
(after the overwriter committed). The list of readers serves as a history of all transactions
that have read $V$. It suffices for $T$ to go though the list for finalizing $\eta$.

In the bitmap-based approach, however, $T$ has to make sure the reader is indeed the predecessor
that read $V$. It is possible that the {\it real} reader $R$ has already left and the same bit position
now points to a totally different transaction, which may or may not have read $V$. As a conservative estimate,
$T$ has to also consult {\tt v.pstamp} to catch such cases after going through all the concurrent readers
using {\tt v.readers} (line~32).
Theoretically, this approach could make $\eta(T)$ larger (hence more false positives) because
a newer reader might update {\tt v.pstamp} with its commit timestamp.
In practice, our evaluation in ERMIA reveals that the benefits of using a bitmap outweighs the drawback.

\subsubsection{Post-commit}
After $\pi$ and $\eta$ are finalized, $T$ tests the exclusion window and aborts 
if necessary (line 34 
in Algorithm~\ref{alg:ssn-parallel-commit}). $T$ then starts the post-commit phase 
to finalize the creation of new versions it wrote
and timestamps of existing versions it read.
As lines 37--42 of Algorithm~\ref{alg:ssn-parallel-commit}
show, $T$ will have to compete with other readers to set {\tt v.pstamp}, so that {\tt v.pstamp} is 
no less than {\tt t.cstamp}. Finalizing {\tt v.sstamp} is straightforward: $T$ simply updates it with
{\tt t.cstamp} and change its type from ``TID'' to ``timestamp'' (line 45). The initialization of new versions
(line 47) is the same as the serial commit protocol described in Algorithm~\ref{alg:ssn-commit}.

\section{Reducing SSN overheads}
\label{sec:reduce-overhead}
SSN requires space and time proportional to the transaction's footprint.
The metadata (e.g., $\pi(T)$ and $\eta(T)$) associated with each version 
and transaction incurs more storage overhead, and post-processing requires 
time proportional to the amount of transaction-private state kept. 
SSN requires pre-commit work proportional to the combined 
size of the read and write sets. 
In particular, the work required for examining the
read set will become a concern for long, read-only and read-mostly 
transactions.

In the rest of this section, we first discuss how SSN can leverage features that
are available in most existing systems to reduce some of the above mentioned overheads.
We then propose two optimizations specifically designed to reduce
the overhead of handling reads.
We adapt the safe snapshot~\cite{pg12} to free read-only transactions from dependency
tracking, and propose an optimization for read-mostly workloads that can avoid
tracking reads of cold data.
With these two optimizations, the vast majority of read-tracking is eliminated,
while serializability is still guaranteed.
The work required at commit time becomes much less for (usual) cases where the write 
set is much smaller than the read set, 
allowing a high-performance implementation of SSN.

\subsection{Leveraging existing infrastructure}
Out of the four machine words SSN maintains in each version, most MVCC implementations 
are already tracking two of them: the version creation timestamp ({\tt v.cstamp}) 
and a pointer to the overwritten version ({\tt v.prev}). 
These are respectively needed to control snapshot visibility and to allow transactions 
to retrieve the appropriate version.
SSN can therefore utilize existing infrastructure, leaving only \texttt{v.pstamp} 
and \texttt{v.sstamp} as new overheads. 
However, we observe that SSN never needs both two values at the same time. 
The \texttt{pstamp} is set at version creation and updated by any reader that commits 
\textit{before} the version is overwritten.
No transaction will access {\tt v.pstamp} once an overwrite of $V$ commits. 
Meanwhile, the overwriting transaction sets \texttt{v.sstamp} when it commits, 
and all subsequent readers will use it to update their own successor timestamps.
We can thus store both fields in a single machine word.  
If the remaining machine word is still objectionable, further space savings could 
be achieved in implementation-specific ways (such as storing a delta that occupies 
fewer bits), but we will not discuss such approaches further here.

As shown respectively by lines 14 and 28 in Algorithms~\ref{alg:ssn-commit} 
and~\ref{alg:ssn-parallel-commit}, a transaction $T$ is considered ``committed'' if 
it survived pre-commit.
The versions $T$ wrote immediately become visible to other transactions (depending on
the underlying CC's visibility policy).
Before $T$ finishes its post-commit phase, readers can use the TID stored in the version
written by $T$ to look up $T$'s status and complete the read if the underlying CC allows; 
the indirection is only used until post-commit converts the TID to a proper
timestamp in each version (described in Section~\ref{sec:implementation}).

\subsection{Safe snapshots and read-only queries}
\label{subsec:safesnap}
In systems that can provide read-only snapshots, including SI-based and
some single-version systems \cite{FOEDUS,silo,hyper}, SSN supports a variant of the 
``safe snapshot'' \cite{pg12}: a transaction 
known to be read-only can avoid the overhead of SSN completely, by using a snapshot 
that the system guarantees will not participate in any serial anomalies. 

The original safe snapshot design was a \textit{passive} mechanism: a query takes a snapshot and 
then waits until all in-flight transactions have ended, while monitoring the system for 
unsafe accesses. If no unsafe accesses occurred before the last in-flight transaction committed, 
the snapshot is deemed ``safe'' and can be used without further concern.
This approach requires tracking transaction footprints after commit, 
and can lead to long delays that make it most suitable for large read-only queries 
executing for tens of seconds or longer.

We instead propose an {\em active} mechanism:
when requested, the system forcibly takes a safe snapshot, with its timestamp 
stored in a global variable. SSN treats the snapshot as a 
transaction that has read every record in the database, inflicting a read
anti-dependency on all update transactions that were in-flight at snapshot creation time.
Update transactions can still overwrite versions in the safe snapshot,
but will abort if they also take a read anti-dependency on a version created before the snapshot.
Reads not using the safe snapshot are unaffected by it.
Simulations suggest that active safe snapshots have minimal impact on abort rate, 
even under heavy contention, unless the time between safe snapshots is less than the expected duration 
of an update transaction (see Sect.~\ref{sec:eval} for details).

Ports and Grittner \cite{pg12} also describe a read-only optimization for SSI that applies 
to SSN over SI: a transaction that enters pre-commit with an empty write set can 
set $c(T)$ to its snapshot time, thus keeping {\tt v.pstamp} smaller and reducing 
(sometimes significantly) the likelihood that a subsequent overwrite will trigger 
an exclusion window violation. Unlike with SSI, however, implementing 
this ``read-only'' optimization with SSN actually 
protects {\em update} transactions from read-only predecessors. 
Protecting writers is more helpful anyway, as read-only transactions would 
normally use the very lightweight safe snapshot mechanism we propose above.

\subsection{Read-mostly transactions}
\label{subsec:read-opt}
SSN relies on version stamps to implicitly track (part of) the dependency graph to test exclusion window violations.
As we have discussed previously, this mandates the tracking of full transactional footprints.
Reads that are not overwritten upon access have to be kept in the transaction's read set
for verification later at pre-commit (line 7 of Algorithm~\ref{alg:ssn-rw});
writes also have to be tracked and get finalized at post-commit.
On the one hand, tracking writes is usually not a concern for OLTP workloads:
compared to the amount of reads a transaction performs,
writes are usually minority, unless the transaction is write-heavy
(in which case, however, it is usually short as well).
On the other hand, emerging heterogeneous workloads feature even more reads \textit{per transaction},
i.e., \textit{read-mostly} transactions \cite{ERMIA}.
It is not uncommon in these workloads that a much longer scan or series of 
reads are mixed with a small but non-empty write set.
Tracking and validating a large read set dominates the cost of SSN because the write set is tiny by comparison.
As we discuss in Sect.~\ref{subsec:eval-read-opt},
examining each of these reads during pre-commit is a major potential source of last 
level cache misses that drags down the system's performance.

It is worth noting that because these read-mostly transactions' read set is 
\textit{much larger} than their write set, it is unlikely for most
reads to be overwritten by concurrent update transactions.
We leverage this fact to optimize read-mostly transactions:
versions that are not overwritten recently (governed by a threshold) are deemed ``stale'' and 
are not tracked in the read set.

Eliminating the tracking of (potentially long) reads reduces the burden of readers
at pre-commit to verify their reads, however, this also brings challenges for
SSN to finalize $\pi$ and $\eta$ for testing exclusion window violations.
First, writers are unable to obtain up-to-date predecessor status because
the readers that skip read tracking will not update the \texttt{pstamp}
of each read version.
Second, it becomes impossible for readers to check whether the read versions are
overwritten by concurrent writers at pre-commit, because they are not tracked at all.
This fact has profound impact on the parallel commit protocol:
the happens-before relationship we relied on (when finalizing $\eta$) will not hold anymore.
In Algorithm~\ref{alg:ssn-parallel-commit} (line 27),
an updater only needs to wait for the commit result of the readers who entered pre-commit earlier;
those who entered pre-commit later than the updater will in turn spin on the updater (successor),
in case the updater has not finished pre-commit yet.
When we skip tracking certain reads, however,
there is no chance for a reader to spin on its successor who entered pre-commit earlier---it 
even does not have a chance to know the existence of such updaters.
Allowing a reader to commit without properly accounting for its successors will potentially lead to non-serializable execution.

To solve these problems, we again leverage the fact that main-memory systems execute
a transaction on a single thread from beginning to the end without migrating among threads.
We allow read-mostly transactions to commit without having to stamp each read version,
but instead require each thread record the {\tt cstamp} of the read-mostly transaction 
on whose behalf the thread has committed.
As shown in Figure~\ref{fig:tx-table}, the reader puts its {\tt cstamp} in the 
\texttt{last\_cstamp} field in its private entry in the translation table upon commit.
The \texttt{last\_cstamp} field essentially serves as a proxy that summarizes the commit
stamps of all the read-mostly transactions on the same thread.
Because of the lack of read tracking, readers will likely to have larger $\pi$ values,
leaving more back edges unaccounted for and becoming more unlikely to be aborted.
Since readers might not update {\tt pstamp}, it then becomes the updater's responsibility
to detect non-serializable schedules.

At pre-commit, the updater examines the readers bitmap and uses set bits to
find the status of each thread. Note that a reader still indicates its existence in
the readers bitmap upon version access, but without clearing it upon conclusion---the read 
is not tracked in the first place.
The updater needs to consult \texttt{last\_cstamp} to find out the last committed reader's
{\tt cstamp} when calculating {\tt t.pstamp}.

We note, however, {\tt t.pstamp} calculated in this way is conservative 
and can admit false positives: a reader that set but did not clear the bit position
in {\tt v.readers} might lead the updater to inspect another completely irrelevant transaction---one
that has non-overlapping footprints---and cause unnecessary aborts.
Suppose transaction $T$ being executed by thread $t1$ read $V$ without tracking it.
$T$ would have set the bit position in {\tt v.readers} and commit by setting $t1$'s last cstamp
to {\tt t.cstamp} without clearing the bit in {\tt v.readers}.
After $T$ committed, transaction $U$ running on thread $t2$ overwrote $V$ and entered pre-commit.
Suppose $t1$ now starts another transaction $R$ whose footprint does not overlap with $T$ or $R$.
During pre-commit, $U$ would follow {\tt v.readers} to find $R$, because it inherited the bit that was
used by a previous reader $T$, although $R$'s footprint does not overlap with $U$'s.
According to Algorithm~\ref{alg:ssn-parallel-commit}, 
depending on the final $\pi$ and $\eta$ values calculated, $U$ might abort unnecessarily.
We expect that such false positives are not a major concern for workloads with significant portions of
long, read-mostly transactions where reads are the majority.

Using a thread-private \texttt{last\_cstamp} as a proxy that accumulates
\texttt{pstamp}s solves only half of the problem: an updater can account for
read-mostly transactions that entered pre-commit earlier, but not those that
entered pre-commit later. Recall that a read-mostly transaction might
\textit{not} spin on its successor (line 20 of
Algorithm~\ref{alg:ssn-parallel-commit}) with a smaller \texttt{cstamp} because
of the lack of read tracking. Consequently, the updater (as a successor) will
then have to figure out the read-mostly transaction's state when it discovered
that a concurrent reader exists through the readers bitmap. Otherwise, the
updater would have to blindly abort, which will make it hard to commit
write-intensive transactions that have overlapped footprint with read-mostly
transactions, especially when the read-mostly transaction is much longer.

Therefore, it would be desirable for the writer to update the read-mostly
transaction's \texttt{sstamp} (using the updater's \texttt{cstamp}) during
pre-commit; the reader proceeds as usual and tests for exclusion window violation
at the end of pre-commit. We employ a lightweight locking mechanism for the
\texttt{sstamp} to guarantee correctness: the most significant bit (MSB) of
\texttt{sstamp} serves as a lock; the \texttt{sstamp} value updated when its MSB
is unset is guaranteed to be taken into account by the reader. The updater
issues a \texttt{CAS} instruction to update the reader's \texttt{sstamp} with its
\texttt{cstamp}, expecting the MSB is 0. The reader should atomically set the MSB
of \texttt{sstamp} to 1 (e.g., by using an atomic \texttt{fetch-and-or}
instruction) right before it tests for exclusion window violation at the end of
pre-commit. An updater that failed the \texttt{CAS} because the MSB is set will
have to abort. In this way, we reduce unnecessary aborts of updaters, although
more heavy read-mostly transactions might be aborted than without this
lightweight machinery. Our empirical evaluation in
Sect.~\ref{subsec:eval-read-opt} reveals that this effect is minimal, and SSN can
still achieve a transaction breakdown that is close to the
specification under a variant of the TPC-E \cite{TPC-E} benchmark.

\section{Locks and phantom avoidance}
\label{sec:phantoms}
The description of SSN in the previous section works with per-transaction read
sets and write sets, with the assumption that these sets contain versions of
records. To ensure full serializability, e.g., repeatable counts, we need phantom
protection, or the ability to prevent insertions that would change the results of
an uncommitted query. Systems based on pessimistic concurrency control employ
several useful concepts such as hierarchical locking and lock escalation to
reduce tracking overheads, plus key and predicate locking
approaches~\cite{EswaranGLT76,gray78,Mohan90} for phantom prevention in ordered
indexes such as B-trees. SSN is compatible with those mechanisms to prevent
phantoms and so guarantee full serializability. To the extent that the underlying
CC implementation is already phantom-free (as is the case for many recent
systems~\cite{ERMIA,JohnsonPHAF09,FOEDUS,silo}), we will not need to redo them in
SSN. Otherwise, the following subsections describe how to incorporate phantom
detection into the SSN protocol.

\subsection{Hierarchical dependency tracking}
We first adapt the idea of hierarchical locking to SSN's dependency tracking
needs. In a traditional lock-based system, the database is organized as a
hierarchy: schemas, tables, pages, and records. Transactions acquire a concrete
lock on the finest-grained object that suits their needs, and {\em intention
locks} on the object's parents in the hierarchy. With a hierarchical locking
scheme in place, lock escalation also becomes possible: a transaction can choose
to replace a large number of fine-grained locks with a single coarse-level lock,
trading off reduced tracking overhead for an increased risk of conflicts.

We can adopt the same philosophy in SSN: a transaction that will read the
majority of a table can acquire a single read ({\tt R}) lock on the table, and
only needs to update the table-level {\tt pstamp}. Meanwhile, updating
transactions acquire {\tt IW} and {\tt W} locks on the table and individual
records, respectively. They update only {\tt v.sstamp} but must check both table-
and version-level {\tt pstamp}s to detect all conflicts.
Table~\ref{tbl:lock-mode-actions} summarizes the pre-commit checks and
post-commit updates required when the system supports intention modes.

If update contention by readers is a concern, either the lock or the
corresponding pseudo-version can be replicated, following the ``super-latching''
feature of SQL Server~\cite{SuperLatch}. A reader (common case) can then update
only one of many sub-versions, with the trade-off that a writer (infrequent) must
examine all of them.

\begin{table}[t]
  \centering
  \begin{tabular}{c c c}
    \hline
      {\bf Mode} & {\bf Check} & {\bf Update}\\
      \hline
      {\tt R} & {\tt V.sstamp} & {\tt R.pstamp} \\
      {\tt IR} & {\tt V.sstamp} & {\tt R.pstamp, V.pstamp} \\
      {\tt IW} & {\tt R.pstamp, V.pstamp} & {\tt V.sstamp} \\
      {\tt W} & {\tt R.pstamp, V.pstamp} & {\tt V.sstamp} \\
      \hline
  \end{tabular}
  \caption{Lock modes and their commit actions.}
  \label{tbl:lock-mode-actions}
\end{table}

\subsection{Predicates and phantoms}
In addition to preventing dependency cycles between reads and writes, a
serializable system must prevent the phantoms that arise if an insertion would
change the result of a scan. In a database with no installed indices, the
hierarchical lock system described above detects all phantoms: any scan---no
matter how selective its predicates---must access the entire table, and the
resulting table {\tt R} lock will conflict with the {\tt IW} locks of both
inserts and updates. However, predicates involving a (partial) index key mean
finer-grained range scans that access only a portion of the table. Phantom
protection can be achieved in these cases by locking the gaps between keys that
fall inside the range being scanned.

Several gap-locking schemes have been proposed~\cite{Mohan90,Lomet93}, and any of
them could be adapted for use with SSN. We describe here a variant of the scheme
due to Graefe~\cite{graefe07}, where each lock is a composite that can
independently reference a particular key and/or the gap that follows that key.
Both keys and gaps can be held in read and write mode, with conflicts tracked in
piecewise fashion. For example, pairing {\tt W/N} (key-write, gap none) with {\tt
N/R} (key none, gap read) does not imply any dependency edge, but {\tt W/N} and
{\tt R/R} implies an edge because both transactions accessed the key (there is no
conflict on the gap). The full action/mode table can be generated mechanically
(component by component) using Table~\ref{tbl:lock-mode-actions} as a starting
point, so we do not reproduce it here.

With locks that cover key/gap pairs, SSN can prevent phantoms without abandoning
the notions of read and write sets: when a transaction inserts into an index, its
write set contains a version for the key (probably its index entry) associated
with either a {\tt W/N} or {\tt N/W} lock, depending on whether the key was
already present. Meanwhile, the read set of a range-scan transaction contains
index entries it read, each associated with {\tt R/N}, {\tt N/R}, or {\tt R/R}
locks, depending on whether key, gap, or both fall within the scan's endpoints.
From there, the normal SSN machinery will see these new ``reads'' and ``writes'',
and check for exclusion window violations. 

\section{SSN in simulation}
\label{sec:simulations}
We implement the SSN protocol from Sect.~\ref{sec:ssn} in a discrete event
simulator\footnote{Code and scripts available at
\url{https://github.com/ermia-db/ssn-simulator}.} to examine SSN's accuracy and
impact over a wide variety of transaction
profiles, contention levels, and schedules. We are especially interested
in the impact of contention, interference among readers and writers in a mixed
workload, and the impact of active safe snapshots on writer abort rates. In the
next section, we implement parallel, latch-free SSN in ERMIA \cite{ERMIA} to
measure actual commit and abort rates with variants of the TPC-C and TPC-E
benchmarks. 

\subsection{Simulation framework}
\label{sec:simulator}
We have implemented in Python a discrete event simulator designed specifically to
evaluate CC schemes. We use it to compare the CC schemes listed in
Sect.~\ref{sec:background}, with and without SSN. The simulator allows us to
quantitatively compare supported concurrency levels and abort rates of different
concurrency control models. It also exposes anomalies that would indicate
potential design flaws. The simulator was invaluable not only in performing
evaluations and isolating bugs in the various models, but also in driving the
discovery and proof of SSN in the first place.

In all models, the simulator serializes write conflicts by blocking; 2PL and RCL
also block readers that conflict with writers. To bound delays and avoid
deadlocks, we apply a variant of wait depth limiting (WDL)~\cite{alex97}: the
system aborts any transaction that attempts to block on a predecessor that has
already blocked. Because all deadlocks necessarily involve transactions blocking
on blocked transactions, using WDL also also obviates the need for deadlock
detection. Under this arrangement, performance degrades much more gracefully
under contention than it would otherwise. This is especially important for 2PL,
which tends to ``freeze up'' once the combined transactional footprint of blocked
transactions encompasses a majority of the working set. Without WDL, the system
can suddenly enter overload and the expected wait times spike upwards by several
factors, increasing the aggregate transactional footprint even further in a
vicious cycle. With WDL, 2PL achieves drastically better performance than is
traditionally reported, both in terms of latency and completion rate. Meanwhile,
the effect on non-locking schemes is minimal. Even under the most extreme
contention, RC---whose failures are all due to WDL---has a commit rate better
than 90\%.

The simulation framework provides basic support for statistics and monitoring,
automatic detection of serial anomalies, scheduling, and multi-versioned data
access. Pluggable database models then implement the specific CC schemes,
including 2PL and RCL (which simply choose not to return overwritten versions).
The base simulator comprises $\sim$1200 LoC, and most models require 200-300
additional LoC (SI and 2PL are extreme cases, at 80 and 400 LoC, respectively). 

To ensure runs using different CC methods are comparable, we use an open queuing
system: each client submits requests at predetermined times (at intervals roughly
equal to expected transaction latency), independent of previous requests. This
models the real world of connection concentrators and users who do not coordinate
with each other before submitting requests. Thus, if two simulations are started
with the same random seed, the same transactions will be offered at precisely the
same times for both, independent of delays imposed by the CC model in use. Thus,
any difference in throughput, relative latency or abort rates is due to the
models themselves, not differences between transactions offered. Further, exact
reproducibility allows standard test case reduction tools to isolate problems
from a large simulation trace.

Finally, we point out one caveat: although the simulator models transaction
execution times, the low-fidelity timing model does not account accurately for
overheads and bottlenecks that would arise in a real system implementing these CC
schemes. We present the results only to show the relative timing and concurrency
characteristics of different CC schemes. A later section presents results for one
implementation of SSN in ERMIA, but an exhaustive performance study of optimized
implementations in multiple database engines is outside the scope of this paper.

\subsection{Microbenchmark description}
Our simulated evaluations use an enhanced version of the SIBENCH~\cite{crf09}
microbenchmark. The database consists of a single table and a fixed number of
records. Each record contains a single attribute that stores the TID of the
transaction that last wrote to it. Each transaction makes a random number of
accesses, selected uniformly at random from a tunable range of valid footprint
sizes. The last $m$ of those accesses are writes (with $m$ being a workload
parameter). Repeated reads, repeated overwrites, and blind writes are all
allowed. Mixed workloads can be emulated by instantiating multiple client groups
with differing parameters (e.g. to mingle large read-only queries with short
write-intensive transactions).

The benchmark logs the w:w or w:r dependencies implied by each access, using the
TID stored in each record to identify the predecessor. After the run completes, a
post-processing step reconstructs the r:w anti-dependency edges and tests the
resulting graph for strongly connected components (SCC). To avoid blaming one
cycle on multiple transactions (and to allow blaming multiple cycles on one
transaction), we only report transactions in an SCC as serialization failures if
they \textit{also} fail the exclusion window test (every SCC is guaranteed to
contain at least one such failure). This pruning strategy is quite effective in
practice, flagging only 1--2 transactions from a typical SCC involving 2-20
transactions, or up to dozen in SCC with hundreds of transactions. Nevertheless,
we recognize that this strategy overestimates the true number of serialization
failures, because SSN admits false positives.

Offline cycle testing is important for two reasons. First, it is completely
independent of the concurrency control mechanism used; the simulator does not
trust {\em any} CC scheme to be correct (with or without SSN). Second, long
chains of r:w anti-dependency edges can produce serialization failures that reach
arbitrarily far back in time, and a test with a limited horizon would fail to
detect such cycles.\footnote{For example, consider $T_1 \wrdep{} T_n \rwdep{}
\cdots T_i \cdots \rwdep{} T_1$, where each $T_i$ begins just before $T_{i-1}$
commits.} 


\begin{figure}[t]
\includegraphics[width=\columnwidth]{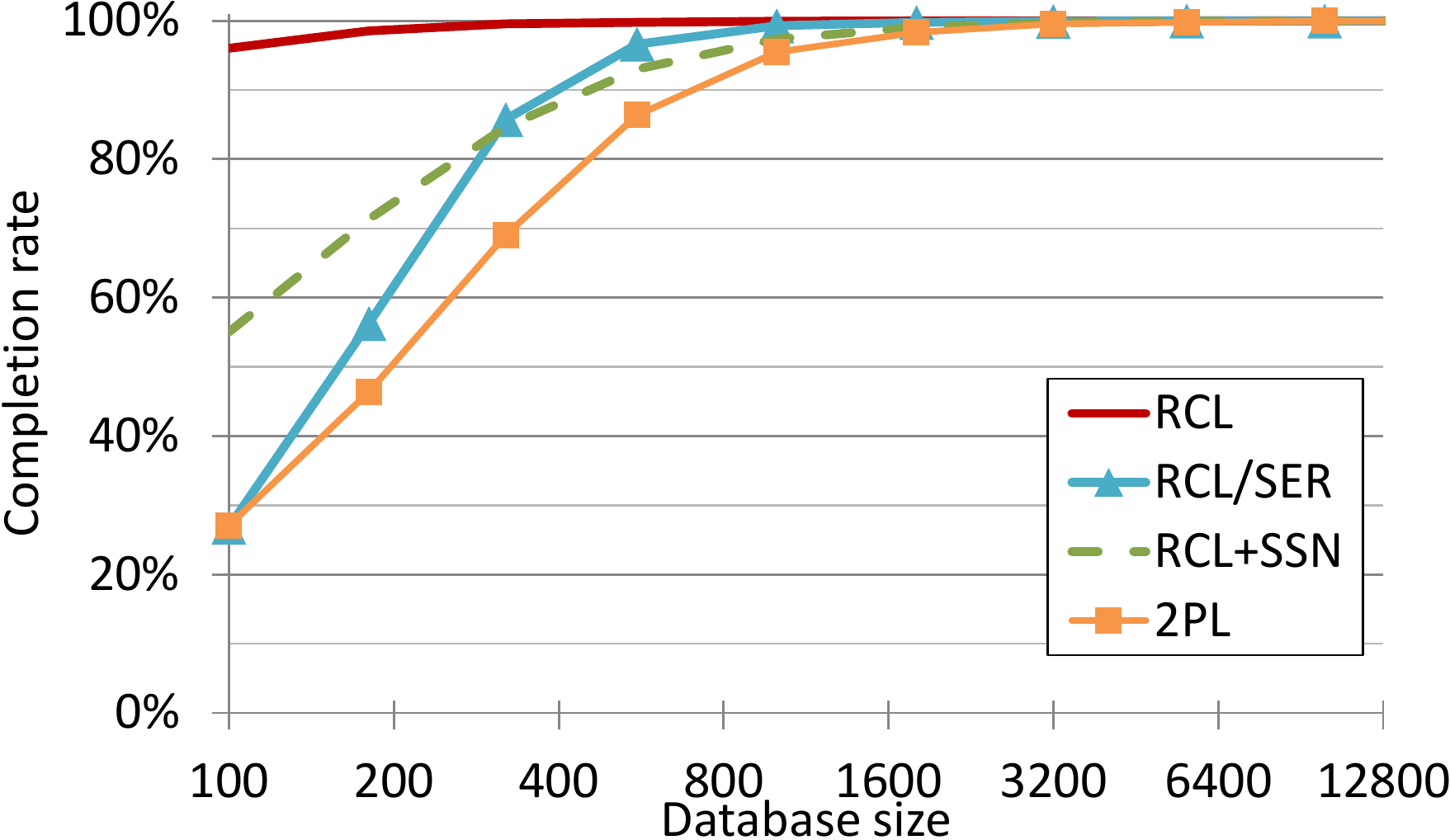}
\includegraphics[width=\columnwidth]{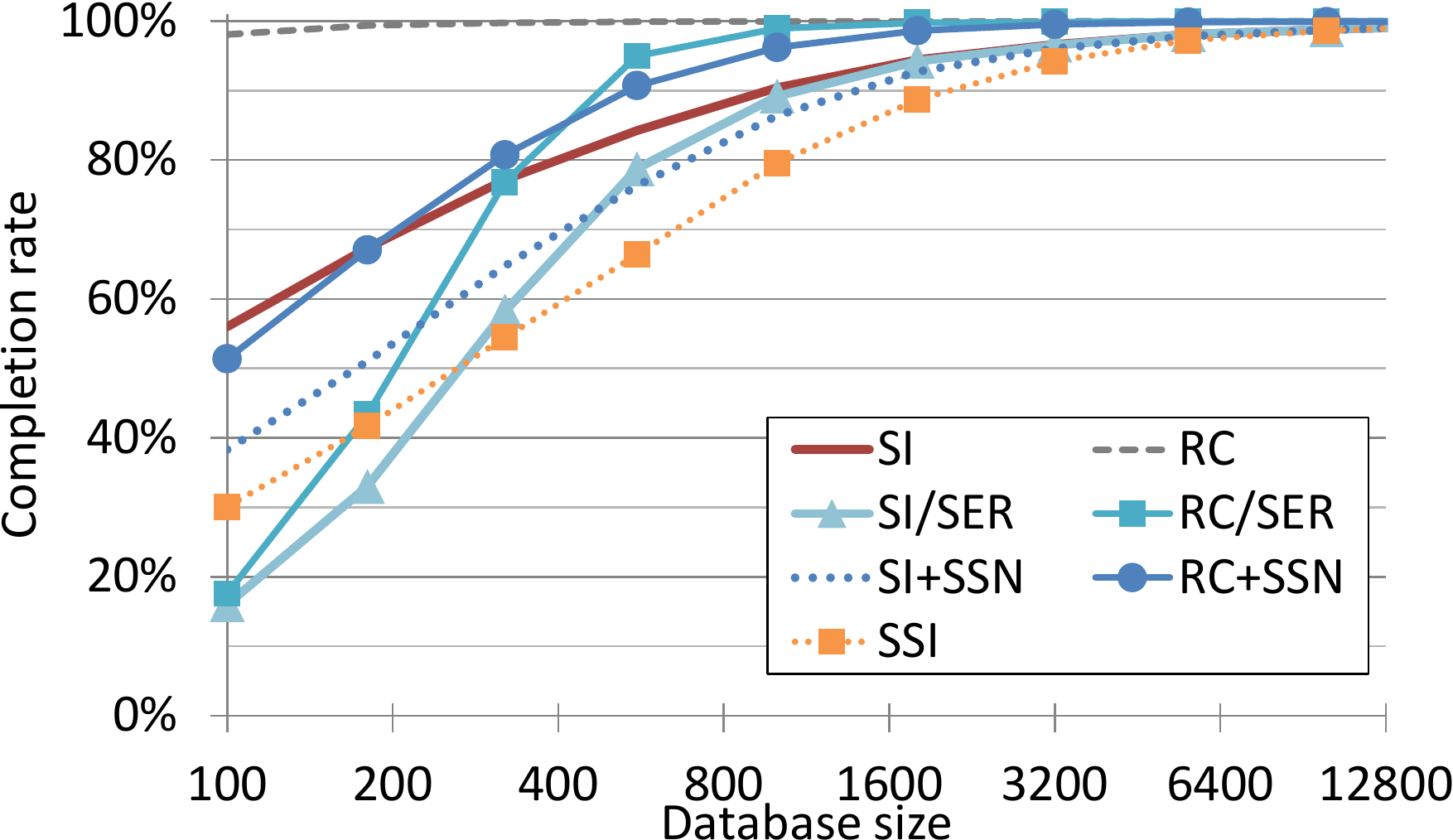}
\caption{The effect of contention for single-version (top) and multi-version (bottom) models,
with 30 clients.}
\label{fig:run-db-size}
\end{figure}

\subsection{SSN and other schemes under contention}
Our first experiment, shown in Fig.~\ref{fig:run-db-size}, calibrates
expectations. Each transaction makes between 8
and 12 accesses, with 25\% of them being writes. We fix the number of clients at
30. The figure shows completion rates of single-version (top) and multi-version
(bottom) CC schemes as the database size varies along the log-scale horizontal axis.
Contention decreases as the
database size increases. For reference only, we show the (non-serializable)
commit rates of SI, RC and RCL. We also show the effective commit rates of those
schemes if we subtract off the number of serialization failures the simulator
reports, given as SI/SER, etc.  Note that the numbers for RCL/SER and SI/SER are
highly unrealistic---requiring offline oracle to compute---and are for reference
only. 
However, the difference between SI/SER and SI, etc, provides information on
how many of the executions are committing with potential anomalies when a
non-serializable scheme is used.

We make two observations: first, SI suffers a lower commit rate because
transactions cannot overwrite versions outside their snapshot. This weakness
extends to SSI and SI+SSN as well. Second, the number of actual serialization
failures remains quite low until contention becomes severe, and gives a sense of
the false positive rates the other schemes produce (quite high for SSI, very low
for RCL+SSN). Finally, we note that protecting weaker CC schemes (RC and RCL)
with SSN yields significantly higher completion rates than any other approach,
across the full range of contention. RCL+SSN, in particular, sees 90\% or
better completion rates until the database size drops below 400 records. SSI
passes that point at 1600 records. Note that this workload averages an
aggregate transactional footprint of 300 records at any given moment. For a
100-record database, RCL+SSN has a completion rate above 50\%, even though three
or more transactions compete for each record. 
A completion rate of above 50\% is
relatively high for this setting,
 considering that every record a transaction
reads will have an active writer with high probability.

Overall, these results indicate that, for these workload parameters, a database
smaller than 500 records suffers severe contention while one larger than 5,000
records is nearly contention-free (though some schemes have non-negligible abort
rates even then). 

\subsection{Transactions with varying write intensity}
\label{subsec:simulation-rw-ratio}
One of the key benefits of multi-version schemes is that reads and writes need
not block each other. A secondary benefit---for read-only queries at least---is
the ability to access a stable snapshot. However, any transaction that makes at
least one update suffers a temporal skew under SI, where all reads occur at the
start of the transaction and all writes take effect at commit time.
Fig.~\ref{fig:run-rw-ratio} illustrates this vulnerability: we run 30 clients,
each making 100 uniformly random accesses against a database containing 100k
records. As the fraction of accesses which are writes increases to 100\% along
the horizontal axis, the SI and non-SI schemes are clearly differentiated, with
the latter all converging to a completion rate nearly doubles that of the
SI-based schemes. The SI schemes all converge to the same performance because
temporal skew is the primary cause of transaction failure. 
In contrast, schemes that always read the latest committed value (2PL, RC, RCL)
are much less vulnerable to temporal skew and consistently achieve better
completion rates. Note that the workload should have low-contention: all clients
together have an aggregate transactional footprint covering at most 3\% of the
database at any given time. 

\begin{figure}[t]
\includegraphics[width=\columnwidth]{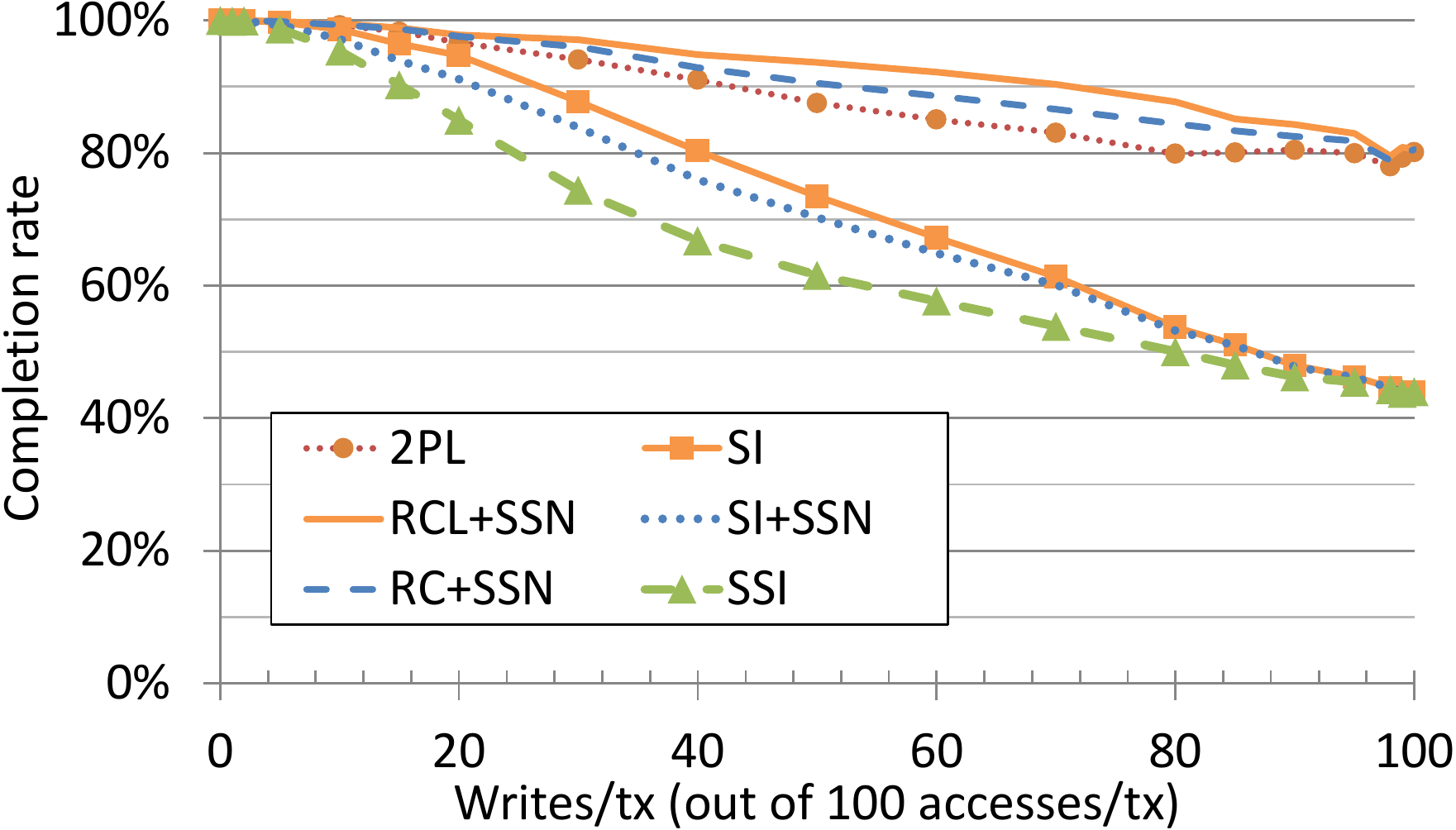}
\caption{Effect of write-intensive transactions.}
\label{fig:run-rw-ratio}
\end{figure}


\begin{figure}[t]
\includegraphics[width=\columnwidth]{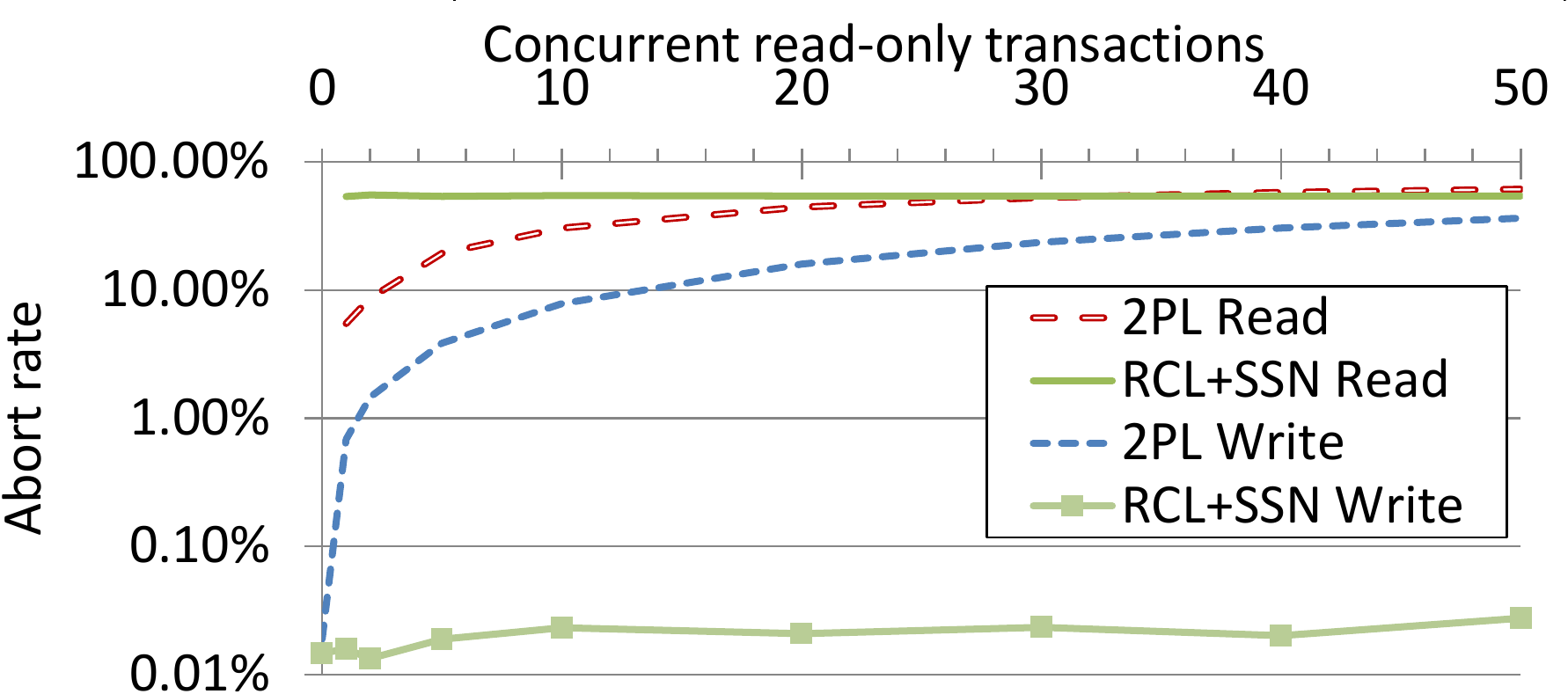}
\includegraphics[width=\columnwidth]{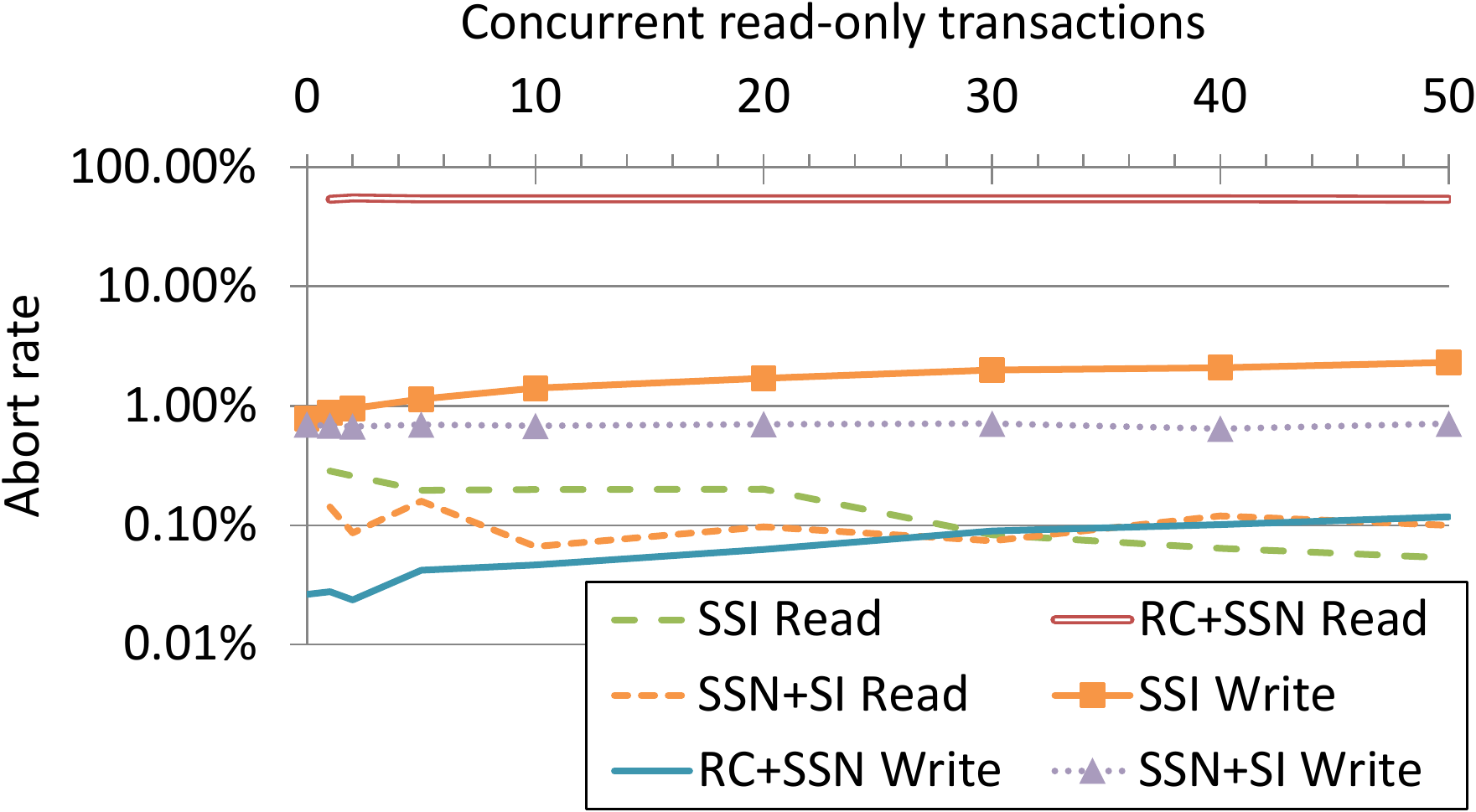}
\caption{Interference between ten update transactions and a varying number of
read-only queries under single (top) and multi (bottom) version models. 
Both RC+SSN and RCL+SSN suffer very high abort rates for readers.}
\label{fig:run-query-suffering}
\end{figure}

\subsection{Interference between readers and writers}
So far, all our simulations involve fairly update-intensive workloads, with
transactions of uniform size and no read-only queries. Lock-based approaches tend
to outperform more optimistic approaches under updates, in no small part because
MVCC is of little use to a writer (who must always overwrite the newest version
of a record). Indeed, we have seen that 2PL performs quite competitively in
update-intensive workloads. However, 2PL interacts very poorly with large
read-only transactions, as demonstrated in Fig.~\ref{fig:run-query-suffering}.
Here, we model a system with 10 update clients (denoted as class ``Write'' in the
figure) and a varying number of read-only clients (denoted as class ``Read'').
Each update client writes between 8 and 12 records, so the aggregate footprint of
the update clients is roughly 100 records.  The database contains 3000 records.
Thus, update clients collectively touch only 3\% of the database at any given
time. We vary along the horizontal axis the number of read-only clients and
measure the resulting abort rate (note the logarithmic vertical axis). Each
read-only client reads between 100 and 200 records (5\% of the database, on
average) before committing. We disable safe snapshots for both SSI and SSN in
this experiment.  This workload exhibits extreme contention under 2PL, with
reader and writer abort rates both quickly approaching 100\% as additional
queries overload the system. 
RC+SSN and RCL+SSN also suffer high abort rates
ranging from 53--55\% for readers across all experiments with readers, 
because the (already-long) query suffers additional
delays due to W-R conflicts that drastically increase the likelihood of a
non-repeatable read that will be aborted by SSN.
In contrast, SI-based models avoid non-repeatable reads, and so achieve completion
rates that suggest low contention: SSI and SI+SSN achieve better than 97\% completion rates
for updates and---thanks to its read-only optimization---99.9\% completion rates
for readers.

\begin{figure}[t]
\includegraphics[width=\columnwidth]{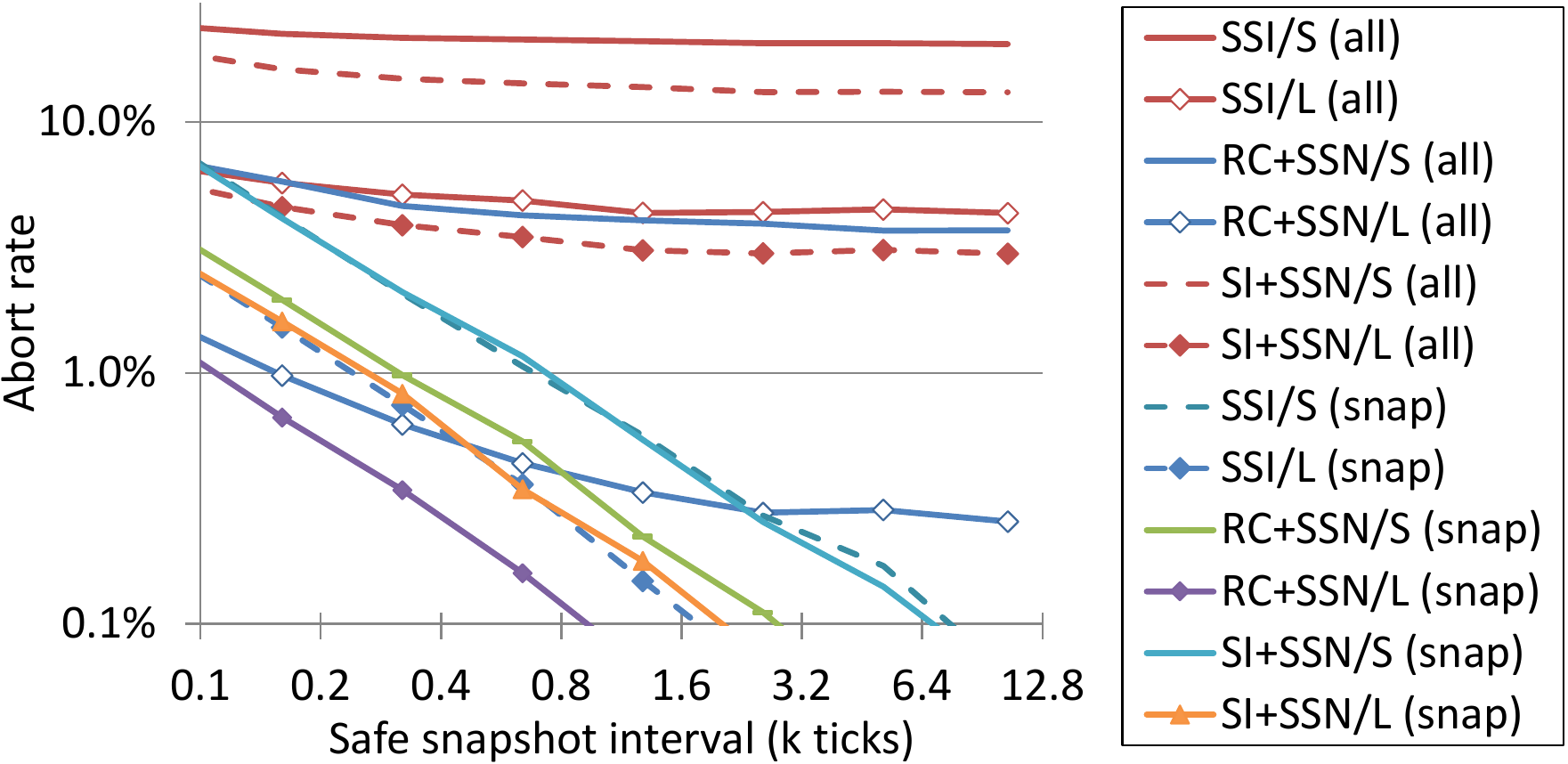}
\caption{Abort rate vs. safe snapshot frequency.}
\label{fig:run-snapkill}
\end{figure}

\subsection{Writer abort rate due to safe snapshots}
Finally, we examine the performance impact of our safe snapshot mechanism. Safe
snapshots forcibly aborts writers that would invalidate a snapshot, so we would
expect higher abort rates in return for reduced latency vs. the passive safe
snapshot described in prior work~\cite{pg12}. Fig.~\ref{fig:run-snapkill}
examines this trade-off, varying the frequency of safe snapshots along the
horizontal axis and plotting the resulting abort rate suffered by two 30-client
update workloads: the class ``S'' workload touches roughly 10 records per
transaction (in a 1000-record database), while the class ``L'' workload touches
40 (in a 4000-record database). Both workloads have a 3:1 read/write ratio, and
scaling the database size with transactional footprint produces similar
contention levels in both. We compare abort rates of SSI, SI+SSN and RC+SSN, for each of
two footprint sizes, differentiating between aborts due to safe snapshot
conflicts vs. other causes. Both horizontal and vertical axes are log-scale.

Even though the workload is rather contentious (aggregate transactional
footprint size is more than 30\% of the database), in most cases abort rates are
relatively low, 10\% or less. The fraction of aborts due to safe snapshot
conflicts drops exponentially as snapshots are taken less frequently. For both
transaction sizes, the snapshot kill rate drops to below 1\% once the delay
between snapshots matches or exceeds the expected update transaction latency
(note the 4$\times$ difference in snapshot interval, corresponding to the
4$\times$ difference in footprint size). Given that most read-only queries are
far larger than any update transaction (the latter tend to finish in a few ms at
most), fairly infrequent snapshots (every 10ms or so) will have virtually little
or no impact on writer abort rates or reader latency.

\section{SSN in action}
\label{sec:eval}
We have incorporated SSN in ERMIA \cite{ERMIA} to provide robust CC for
heterogeneous workloads.\footnote{Code available at
\url{https://github.com/ermia-db/ermia}.} ERMIA is a multi-version,
memory-optimized database system that uses a single atomic \texttt{fetch-and-add}
instruction per transaction to provide cheap global commit ordering, making it
amenable to various CC schemes, such as SI (with and without SSN) and SSI. ERMIA
prevents phantoms at low cost using its index (Masstree \cite{Masstree}). The SSI
implementation in ERMIA follows the parallel commit paradigm described in
Sect.~\ref{subsec:parallel-commit}. In this section, we focus on evaluating the
following:

\begin{itemize}
\item Performance of SSN and other comparing CC schemes under traditional OLTP
workloads (Sect.~\ref{subsec:eval-scale});
\item Impact of the optimizations for read-mostly transactions on heterogeneous
workloads (Sect.~\ref{subsec:eval-read-opt});
\item Effectiveness SSN's safe retry property and SSN's accuracy under high
contention (Sect.~\ref{subsec:eval-retry}).
\end{itemize}

\subsection{Benchmarks}
\label{subsec:eval-bench}
ERMIA implements a wide variety of benchmarks, including TPC-C~\cite{TPC-C},
TPC-E~\cite{TPC-E} and their extensions for different evaluation purposes. We
evaluate SSN and compare its performance with SI, SSI, and optimistic CC
(OCC)~\cite{KungR81} using these benchmarks available in ERMIA. We first use
TPC-C to explore how SSN performs for traditional OLTP workloads with low
contention. We also compare different CC schemes using TPC-CC, a more contentious
variant of TPC-C~\cite{ERMIA}. Finally, TPC-EH, a heterogeneous OLTP workload
(detailed in \cite{ERMIA}) that features long, read-mostly transactions is used
to evaluate the effectiveness of SSN's read optimization. Details of these
benchmarks are described below.

\textbf{TPC-C.}
The TPC-C benchmark simulates an order-entry environment and is the dominant
benchmark for traditional OLTP systems. It is a write-intensive, easily
partitionable and low-contention workload. We partition the database by
warehouse. Each thread is assigned a home warehouse; 15\% and 1\% of the Payment
and New-Order transactions are cross-partition, respectively. We run TPC-C to
compare the performance of different CC schemes under low contention.

\textbf{TPC-CC.}
As we have discussed above, the stock TPC-C benchmark exhibits low contention. To
evaluate SSN under high contention, we use TPC-CC, a variant of TPC-C implemented
in ERMIA that uses a random warehouse for each transaction~\cite{ERMIA}. Instead
of assigning each thread a home warehouse, a thread chooses a warehouse randomly
as its home warehouse upon starting a transaction. The percentage of remote
transactions for Payment and New-Order remain the same as in TPC-C.

\begin{figure*}[t]
\centering
\begin{subfigure}[t]{0.325\textwidth}
\includegraphics[width=\textwidth]{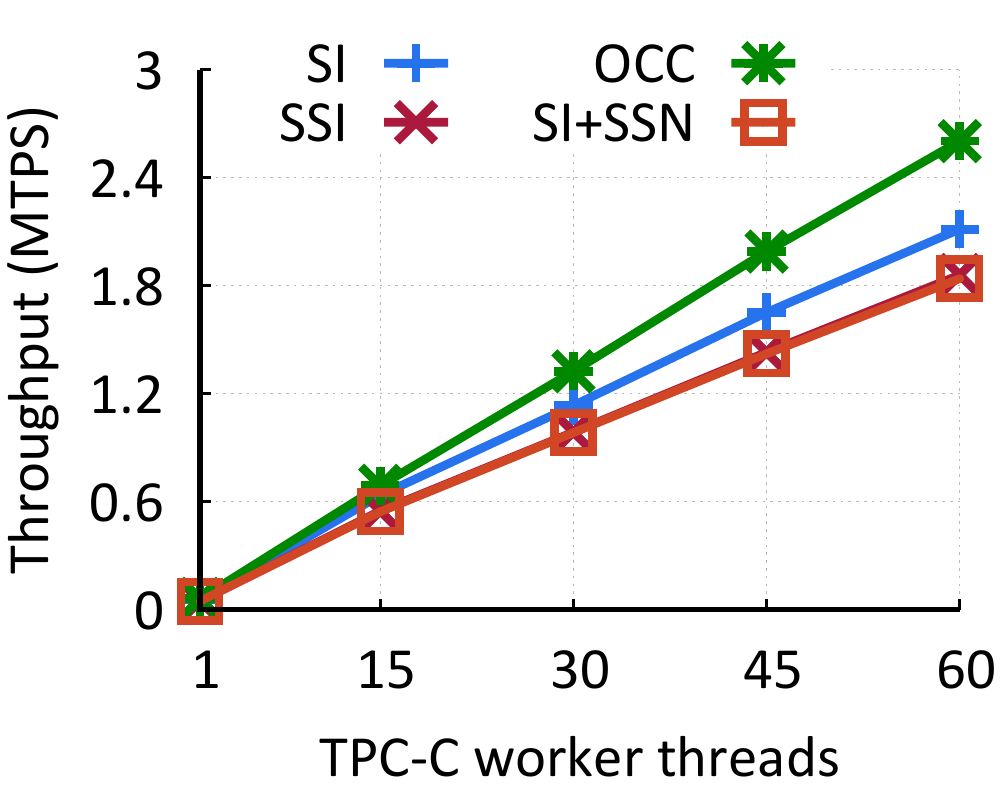}
\end{subfigure}\hfill
\begin{subfigure}[t]{0.325\textwidth}
\includegraphics[width=\textwidth]{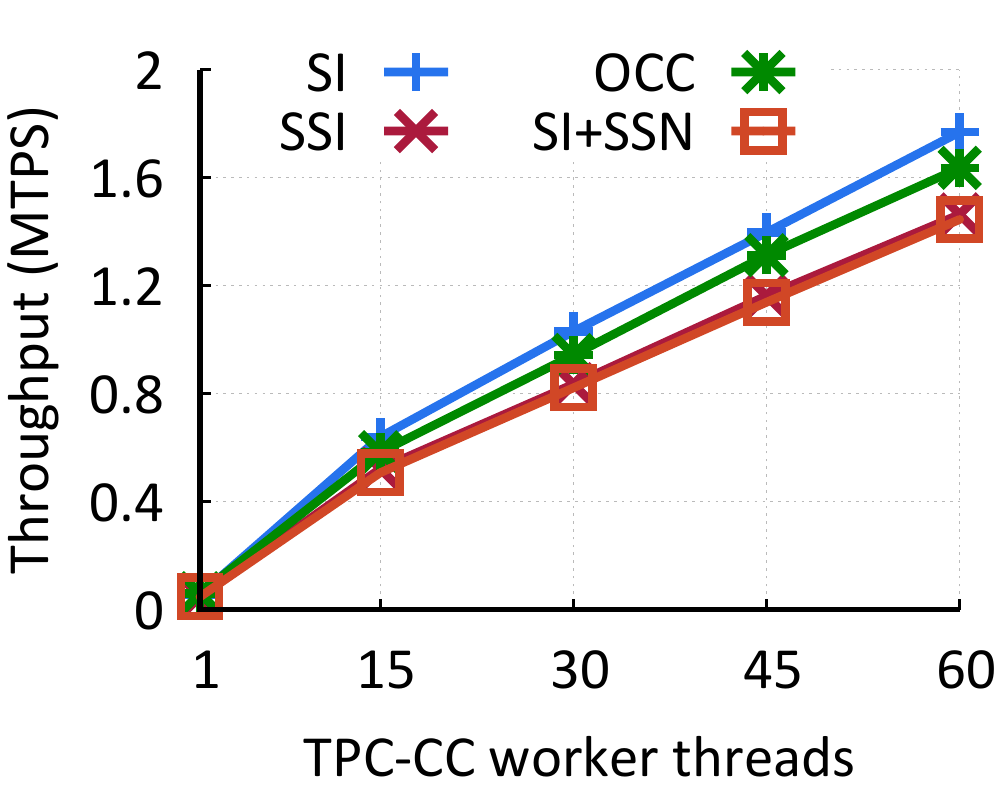}
\end{subfigure}\hfill
\begin{subfigure}[t]{0.325\textwidth}
\includegraphics[width=\textwidth]{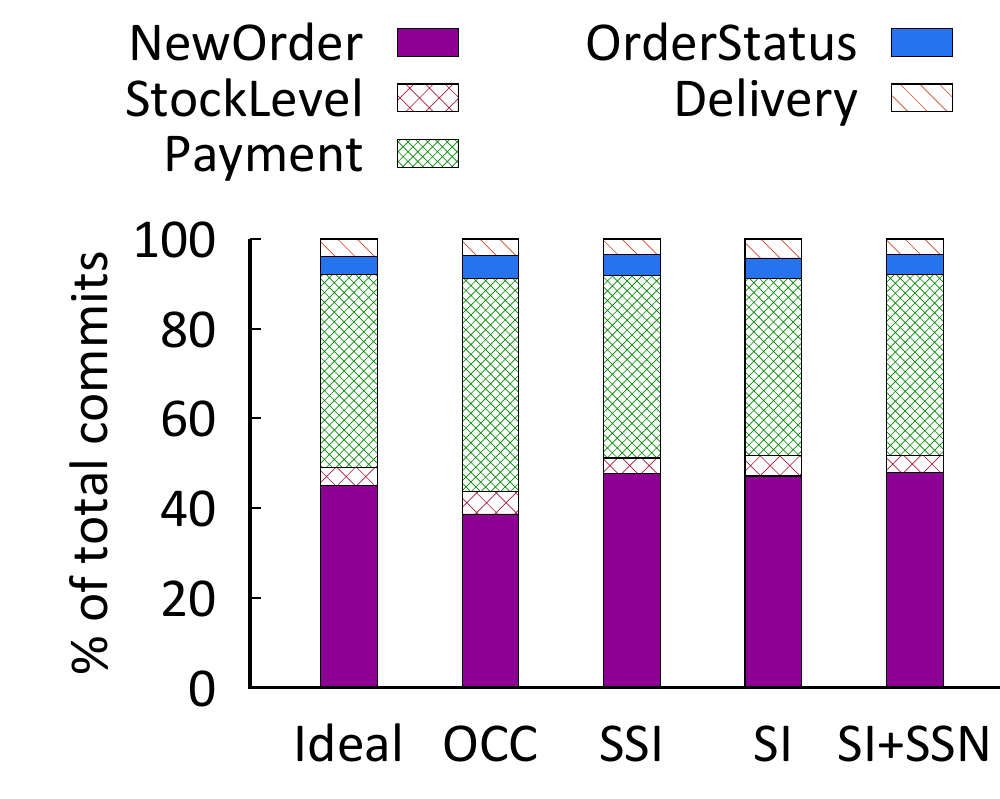}
\end{subfigure}
\caption{Commit throughput of TPC-C (left), TPC-CC (middle), and throughput
breakdown of TPC-CC running at 60 threads (right). Under low contention (TPC-C),
OCC outperforms all the other schemes. The gap between OCC and SSI/SSN shrinks
with more contention (TPC-CC). OCC favors write intensive transactions; other
schemes' profiles are similar to the Ideal.}
\label{fig:tpcc}
\end{figure*}

\textbf{TPC-EH.}
Compared to TPC-C, TPC-E \cite{TPC-E} is a more recent OLTP benchmark that
features more sophisticated and realistic tasks that are performed by brokerage
firms. It has a significantly higher read-to-write ratio ($\sim$10:1 vs.
$\sim$3:1 of TPC-C) \cite{TPC-EC-Comparison}. Although TPC-E models modern OLTP
workloads more realistically, it lacks the support for emerging heterogeneous
workloads, where the execution of \textit{long and read-mostly} transactions are
of paramount importance. TPC-EH \cite{ERMIA} fills this gap by introducing an
additional read-mostly transaction---Asset-Eval---to TPC-E, 
and extending the schema with an Asset-History table. Asset-Eval
aggregates assets for a set of customers and inserts the results to
Asset-History.
 For each customer account, Asset-Eval computes the total
asset by joining the Holding-Summary and Last-Trade tables. As a result,
Asset-Eval will contend mostly with the Market-Feed and Trade-Result
transactions, which modify the Last-Trade and Holding-Summary tables,
respectively. In our experiments, Asset-Eval scans 20\% of all the records in the
Customer-Account table.

The Asset-Eval transaction takes 20\% of the total transaction mix in TPC-EH. 
Because our goal is
to evaluate CC schemes under contention, Data-Maintenance and Trade-Cleanup are
omitted from our TPC-EH implementation.
The
revised transaction mix therefore becomes: Broker-Volume (4.9\%),
Customer-Position (8\%), Market-Feed (1\%), Market-Watch (13\%), Security-Detail
(14\%), Trade-Lookup (8\%), Trade-Order (10.1\%), Trade-Result (10\%),
Trade-Status (9\%), Trade-Update (2\%) and Asset-Eval (20\%). 

\subsection{Experimental setup}
We apply SSN over SI (denoted as SI+SSN) and compare it with other CC schemes,
including SI and SSI in ERMIA, and also with OCC. 
The OCC implementation used in our experiments is
Silo~\cite{silo}, a single-version, main-memory optimized system that uses a
decentralized architecture to avoid physical contention. There have been newer
systems that follow a similar philosophy to achieve even better performance, such
as FOEDUS~\cite{FOEDUS}. However, ERMIA shares the same benchmark code and
implementation paradigm with Silo (e.g., both use threads---instead of
\textit{processes} in FOEDUS---as transaction workers). Therefore, for fair
comparison, we use Silo in our experiments. The version of Silo we used is
augmented with the same TPC-C, TPC-CC and TPC-EH benchmarks in ERMIA.

We run the benchmarks described in Sect.~\ref{subsec:eval-bench} in Silo and
ERMIA under various CC schemes on a quad-socket Linux server with four Intel Xeon
E7-4890 v2 processors clocked at 2.8GHz (60 physical cores in total) and 3TB of
main memory. Each worker thread is pinned to a physical core. We keep all the
data in memory and direct log writes to \texttt{/dev/null}.

The performance numbers we report are averages of three consecutive 10-second
runs, each starting with a freshly loaded database. Unless explicitly stated, all
transactions aborted due to CC reasons (e.g., phantoms, exclusion window
violations and write-write conflicts) are dropped. In production environments,
these aborted transactions should be retried until they successfully commit. We
only avoid retrying to evaluate the fairness among transactions under different
CC schemes. User-instructed aborts (such as those found in TPC-E) are never
retried. For TPC-C and TPC-CC, the number of concurrent threads is fixed to the
scale factor (i.e., number of warehouses) unless otherwise stated. We use ten
working days and a scale factor of 500 for TPC-EH.

\begin{figure*}[t]
\centering
\begin{subfigure}[t]{0.325\textwidth}
\includegraphics[width=\textwidth]{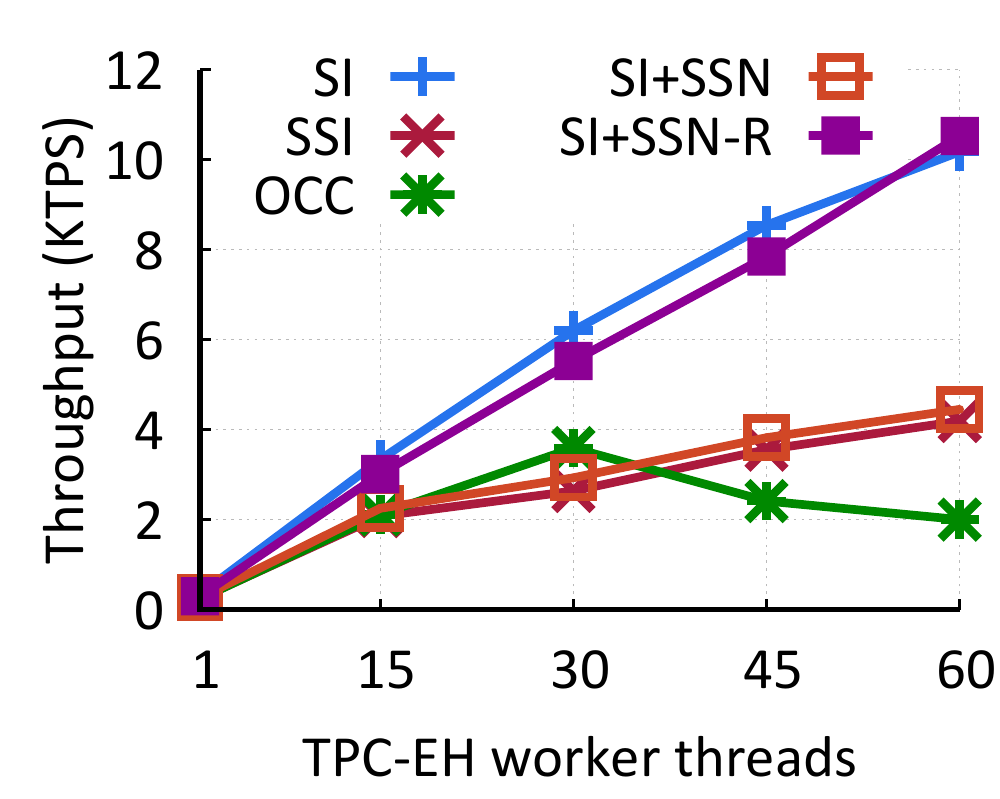}
\end{subfigure}\hfill
\begin{subfigure}[t]{0.6\textwidth}
\includegraphics[width=\textwidth]{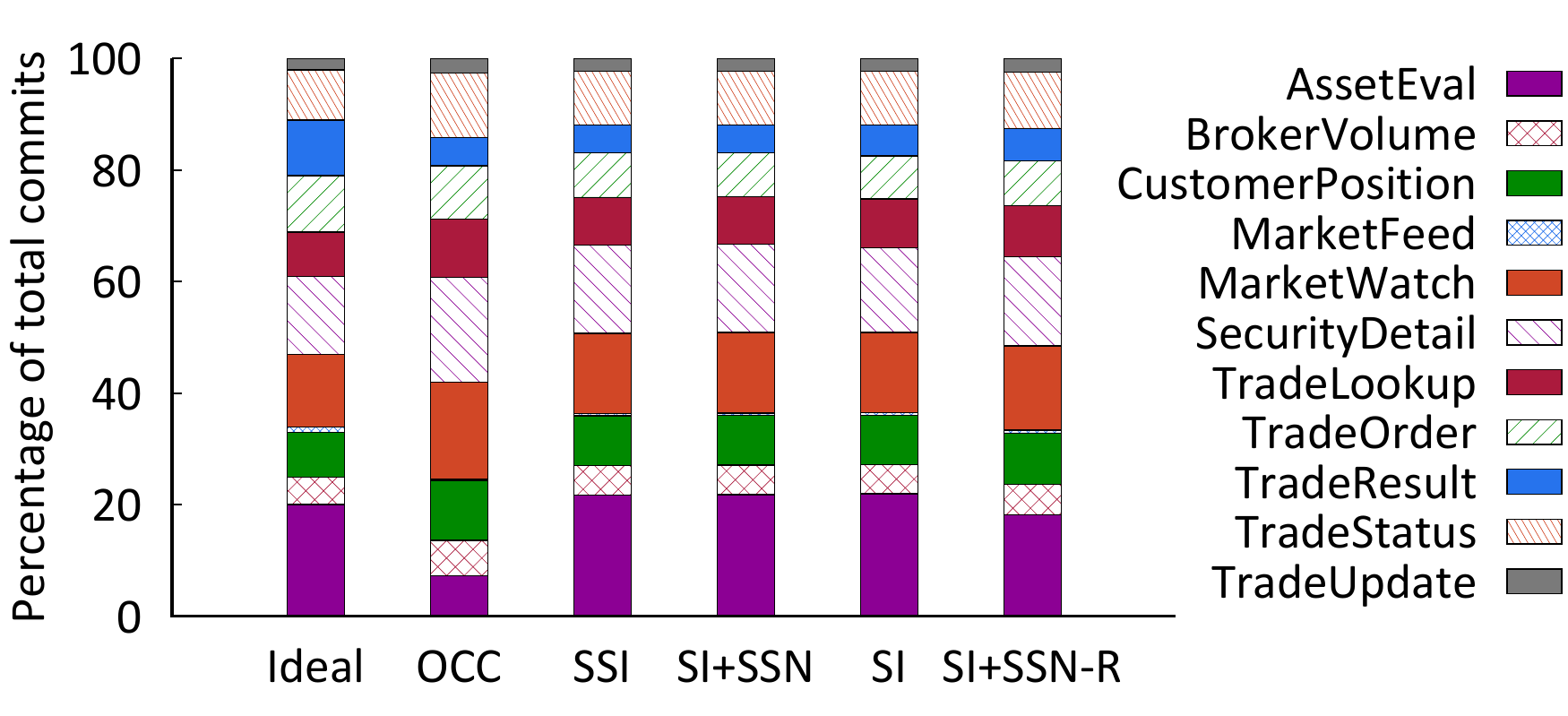}
\end{subfigure}\hfill
\caption{Commit throughput of the TPC-EH benchmark (left) and the throughput
breakdown under 60 threads (right).}
\label{fig:tpceh}
\end{figure*}

\subsection{Traditional OLTP workloads}
\label{subsec:eval-scale}
We first explore how SSN and other CC schemes perform under traditional OLTP
workloads, by comparing the throughput of TPC-C and TPC-CC under various CC
schemes. Fig.~\ref{fig:tpcc} shows the throughput of TPC-C (left) and TPC-CC
(middle) with a varying number of concurrent threads. The number of warehouses is
fixed to the number of concurrent threads. Note that with such setting, neither
TPC-C nor TPC-CC generates enough conflicts to stress the CC significantly.
Therefore, the purpose of this experiment is to understand how different CC
schemes perform for the most common and simple workloads. We explore how they
behave under more contention in Section~\ref{subsec:eval-retry}. SI outperforms
SI+SSN and SSI in all cases, however, it is not serializable. OCC outperforms the
other schemes under TPC-C, which has low contention. With random warehouse
selection in TPC-CC, the gap shrinks and OCC starts to perform similarly to SI.
SSI performs slightly worse than SI+SSN. OCC only marginally outperforms SI+SSN
under TPC-CC, showing the minimal overhead of SSN on top of the underlying CC
scheme.

To further understand how different types of transactions perform under SSN, the
vertical axis of Fig.~\ref{fig:tpcc}(right) presents the relative percentage of
each transaction's commit in the TPC-CC mix, for the different CC schemes in the
horizontal axis, including the transaction mix specified by the TPC-C
specification~\cite{TPC-C} (``Ideal'') for comparison. The experiment was
conducted with 60 threads and aborted transactions are dropped to show any bias a
CC scheme might have toward certain types of transactions. Among all the schemes
we evaluated, OCC has shown a bias toward the write-intensive Payment
transaction, but the other multi-version schemes have shown similar profiles to
Ideal. While this is expected as OCC is known to favor write-intensive
transactions, we emphasize that SI+SSN provides fair scheduling and has kept a
low abort rate, without deviating much from the workload specification. SSN does
not aggravate the underlying CC's bias. We further explore the behaviors of
different CC schemes under high contention in Sect.~\ref{subsec:eval-retry}.

\begin{figure*}[t]
\centering
\minipage{0.32\textwidth}
\begin{subfigure}[t]{\textwidth}
\includegraphics[width=\textwidth]{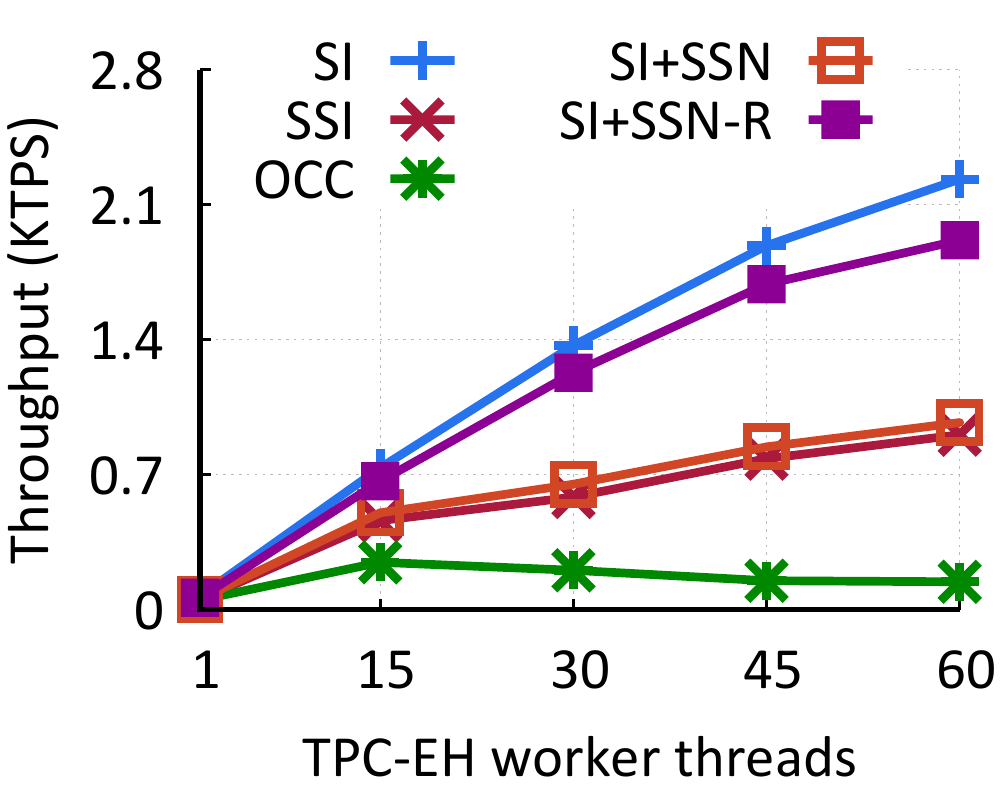}
\end{subfigure}\hfill
\caption{Commit throughput of the Asset-Eval transaction in TPC-EH.}
\label{fig:read-opt}
\endminipage\hfill
\minipage{0.655\textwidth}
\hfill\begin{subfigure}[t]{0.49\textwidth}
\includegraphics[width=\textwidth]{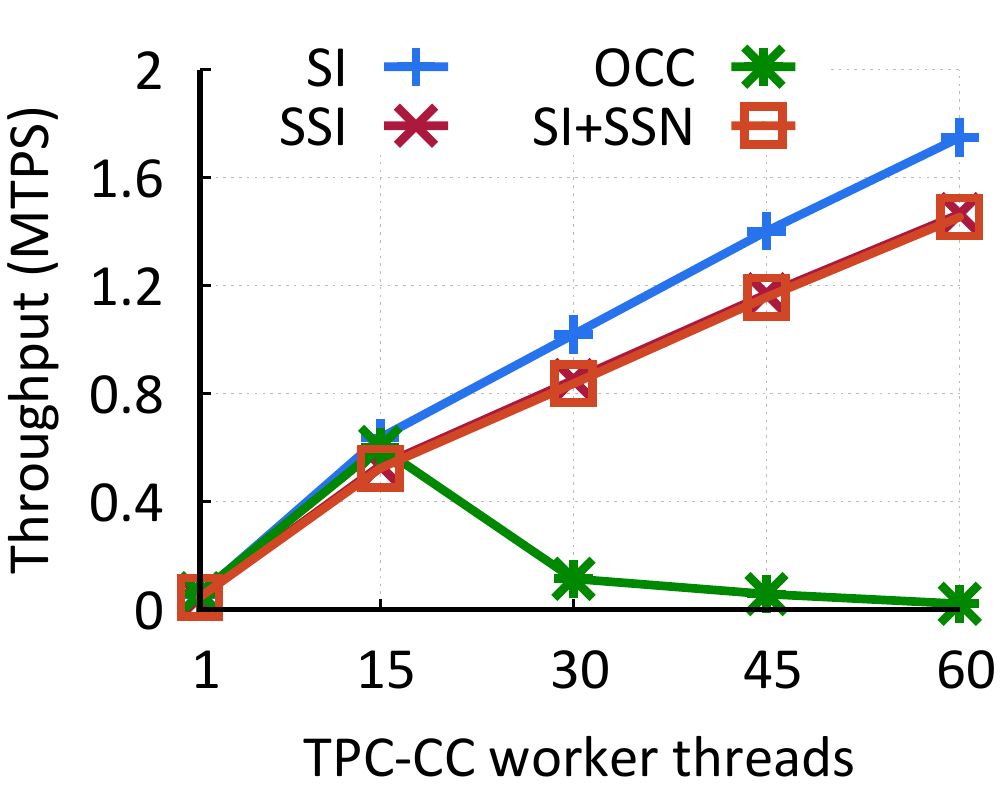}
\end{subfigure}
\begin{subfigure}[t]{0.49\textwidth}
\includegraphics[width=\textwidth]{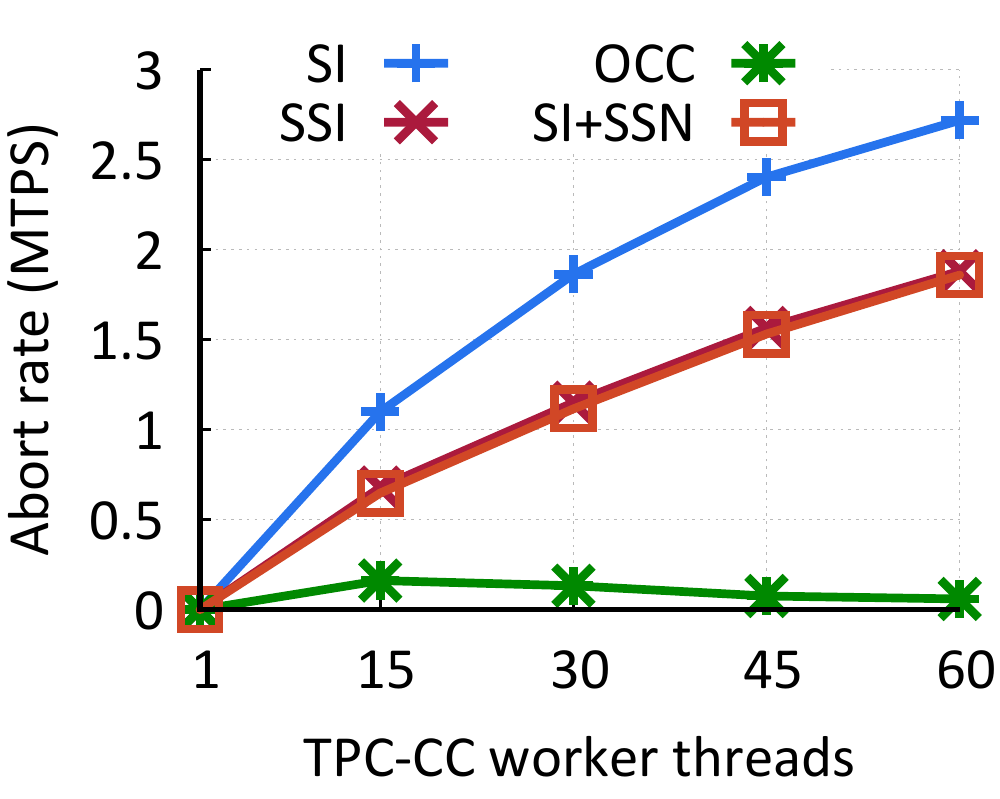}
\end{subfigure}
\caption{Commit (left) and abort (right) rates of TPC-CC. The number of
warehouses is fixed to the number of worker threads.  Aborted transactions are
retried.}
\label{fig:tpcc_contention-retry}
\endminipage
\end{figure*}

\subsection{Read-mostly transactions}
\label{subsec:eval-read-opt}
We evaluate the performance of read-mostly transactions using TPC-EH. As shown by
Fig.~\ref{fig:tpceh}(left), OCC keeps up with the other multi-version schemes
until 30 threads. With more threads, OCC's performance drops sharply, achieving
less than 20\% of SI's throughput at 60 threads. Because of its optimistic and
single-version nature, OCC does not allow any back edges in the dependency graph.
Long reads can be easily invalidated by concurrent, conflicting writers, leading
to massive aborts of read-mostly transactions. We plot the throughput of the
Asset-Eval transaction in Fig.~\ref{fig:read-opt}. In the figure, OCC showed a
declining trend after 15 threads. Although the aggregate throughput kept
increasing from 15 to 30 threads as shown in Fig.~\ref{fig:tpceh}(left), the
corresponding throughput numbers for Asset-Eval transactions in
Fig.~\ref{fig:read-opt} do not show a similar trend. In other words, OCC
processed more other transactions, and adding more workers does not help in
committing more Asset-Eval transactions.  Most of them are aborted by concurrent
updaters to the same scanned region. The throughput breakdown shown on the right
side of Fig.~\ref{fig:tpceh} aligns with this observation: OCC commits much fewer
Asset-Eval transaction than the other schemes. SI+SSN and SSI slightly deviates
from Ideal by committing $\sim$1.76\% more and $\sim$0.8\% fewer of Asset-Eval
transactions, respectively. Unlike single-version OCC, multi-versioning allows
SI-based schemes to accept all (SI) or some (deemed not harmful by SSI or SI+SSN)
back edges in the dependency graph, thus allowing more valid schedules.

As shown by Fig.~\ref{fig:tpceh}, both SI+SSN and SSI scale well under TPC-EH,
but with a widening gap between them and SI as the number of
worker threads increases.
 A similar trend is found for the Asset-Eval transaction in Fig.~\ref{fig:read-opt}.
Compared to SI, SI+SSN and SSI have to maintain a read set in each transaction
for validation at pre-commit. With more concurrent threads, the tracking and
checking of read sets imposes higher overhead. Specifically, as shown by
Algorithm~\ref{alg:ssn-parallel-commit}, SSN has to iterate over the whole read
set during pre-commit (SSI does so, too), which is a major source of last level
cache (LLC) misses. Our profiling results for a 20-second run of TPC-EH under
30 threads show that although SI+SSN's parallel commit procedure only takes
12\% of the total CPU cycles, the function alone incurs 36.25\% and 16.23\% of
LLC load and store misses, respectively. In this experiment, we have added a
new variant (SI+SSN-R) that employs the read-mostly optimizations
(Section~\ref{subsec:read-opt}). SI+SSN-R skips tracking the majority of reads,
thus avoiding most LLC misses during pre-commit. As shown in
Fig.~\ref{fig:tpceh} and~\ref{fig:read-opt}, SI+SSN-R achieves $\sim$136\% and
$\sim$97\% better performance compared to vanilla SSN, for the overall and
Asset-Eval performance of TPC-EH, respectively.

\begin{table}[t]
\setlength{\tabcolsep}{5pt}
\begin{tabular}{l|c|c|c|c}
\bf Transaction   & \bf SI   & \bf SSI & \bf SI+SSN & \bf SI+SSN-R\\\hline
Asset-Eval        & 2232.50  & 907.60  & 970.67     & 1915.35\\\hline
Broker-Volume     & 532.83   & 222.27  & 234.88     & 572.15\\\hline
Customer-Position & 899.02   & 376.98  & 399.32     & 971.44\\\hline
Market-Feed       & 49.10    & 16.68   & 16.58      & 54.92\\\hline
Market-Watch      & 1464.58  & 603.80  & 643.42     & 1588.25\\\hline
Security-Detail   & 1552.43  & 661.19  & 702.25     & 1687.43\\\hline
Trade-Lookup      & 885.77   & 359.61  & 381.04     & 959.91\\\hline
Trade-Order       & 783.42   & 333.41  & 351.53     & 851.80\\\hline
Trade-Result      & 563.76   & 211.33  & 220.02     & 607.78\\\hline
Trade-Status      & 984.51   & 405.26  & 429.60     & 1068.06\\\hline
Trade-Update      & 229.45   & 93.63   & 101.59     & 248.44\\\hline\hline
Total TPS         & 10177.37 & 4191.75 & 4450.90    & 10525.52
\end{tabular}
\caption{Throughput (TPS) of individual TPC-EH transactions under multi-version
CC schemes with 60 threads.}
\label{tbl:tpce20-breakdown}
\end{table}

Astute readers might have noticed that SI+SSN-R could even outperform SI in terms
of total commit rate as shown in Fig.~\ref{fig:tpceh}, especially when running at
60 threads. Table~\ref{tbl:tpce20-breakdown} lists the commit rates of individual
TPC-EH transactions under different CC schemes running at 60 threads. Compared to
SI, SI+SSN-R commits fewer heavy-weighted Asset-Eval and Market-Feed
transactions, leaving more resources available for other transactions. We note
that it is critical to set an appropriate threshold for SI+SSN-R, which governs
whether a version is tracked in the read set. In the case of TPC-EH, our
experiments show that using a low threshold tends to kill more Market-Feed
transactions because of conflicts in the Last-Trade table (reads from Asset-Eval,
updates from Market-Feed), leading to even higher commit rates for other
transactions. In general, under SI+SSN-R an updater that overwrote a stale
version needs to adjust the reader's \texttt{sstamp}, and the longer the reader
is, the easier can an updater catch the reader ``alive'' during the latter's
pre-commit phase. As a result, longer readers will have a higher chance of being
stamped a lower \texttt{sstamp}, making it easier to violate an exclusion window.
For the experimental results reported in this paper, we set the threshold to
\texttt{0xFFFFF}. It provided a balance between overall commit rate and fairness
among individual transactions. With a threshold of \texttt{0xFFFFF}, versions
that are not updated during the past period in which $\sim$1MB of data are
written in the database, are considered stale and consequently not tracked in the
read set.\footnote{Our current implementation accepts a user-defined,
workload-specific threshold. Self-tuning it as workload changes is future work.}
For SI+SSN-R to provide both high aggregate throughput and fair scheduling, one
must adjust the threshold depending on the workload.  In summary, as shown in
Table~\ref{tbl:tpce20-breakdown}, SI+SSN roughly follows SI's breakdown but
provides lower commit rates, while SI+SSN-R sacrifices a little fairness toward
readers but also maintains high aggregate throughput (similar to SI's) with a
proper threshold.

\begin{figure*}[t]
\centering
\minipage{0.32\textwidth}
\begin{subfigure}[t]{\textwidth}
\includegraphics[width=\textwidth]{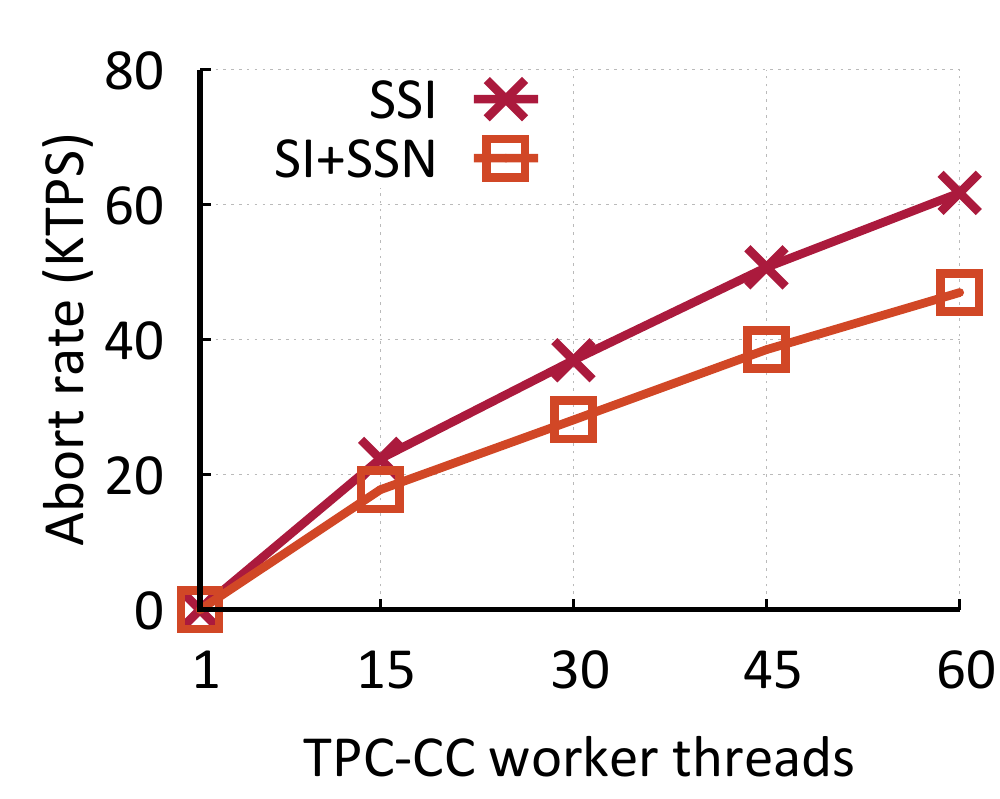}
\end{subfigure}\hfill
\caption{Abort rate of TPC-CC with retry due to SSI/SSN certification failures.}
\label{fig:tpcc_contention-serial-abort}
\endminipage\hfill
\minipage{0.655\textwidth}
\hfill\begin{subfigure}[t]{0.49\textwidth}
\includegraphics[width=\textwidth]{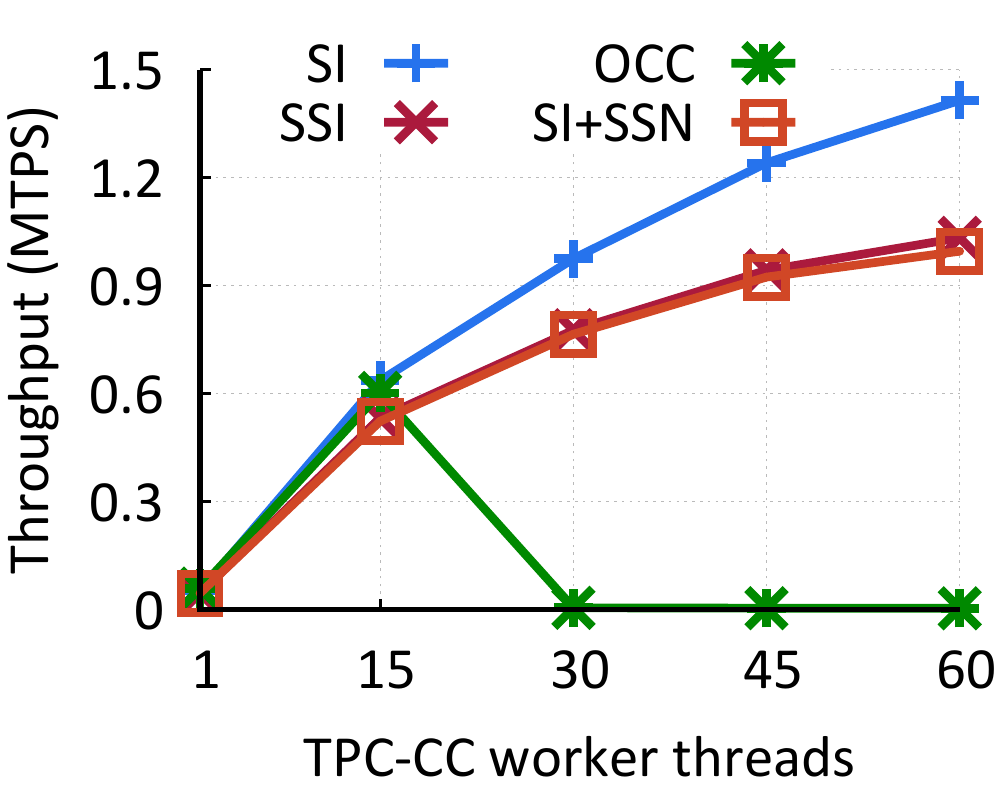}
\end{subfigure}
\begin{subfigure}[t]{0.49\textwidth}
\includegraphics[width=\textwidth]{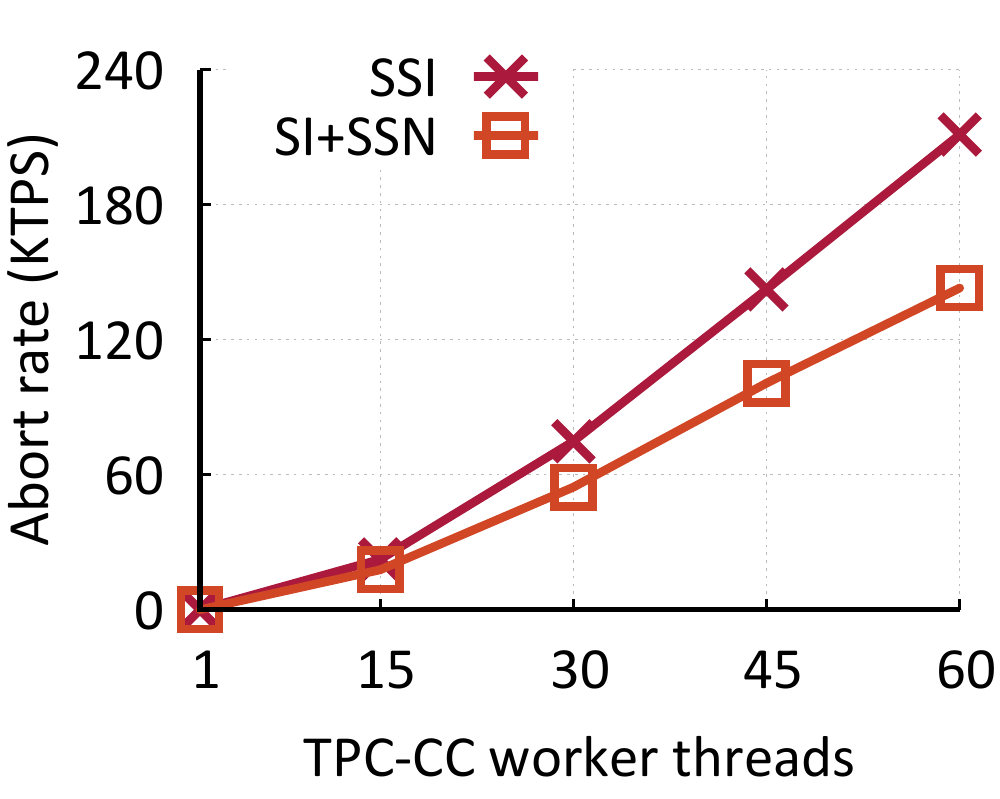}
\end{subfigure}
\caption{TPC-CC throughput (left) and abort rate due to certification failures
(right) with fixed database size (15 warehouses). Aborted transactions are
retried.}
\label{fig:tpcc_contention-15wh-retry}
\endminipage
\end{figure*}

\begin{figure}
\includegraphics[width=0.49\textwidth]{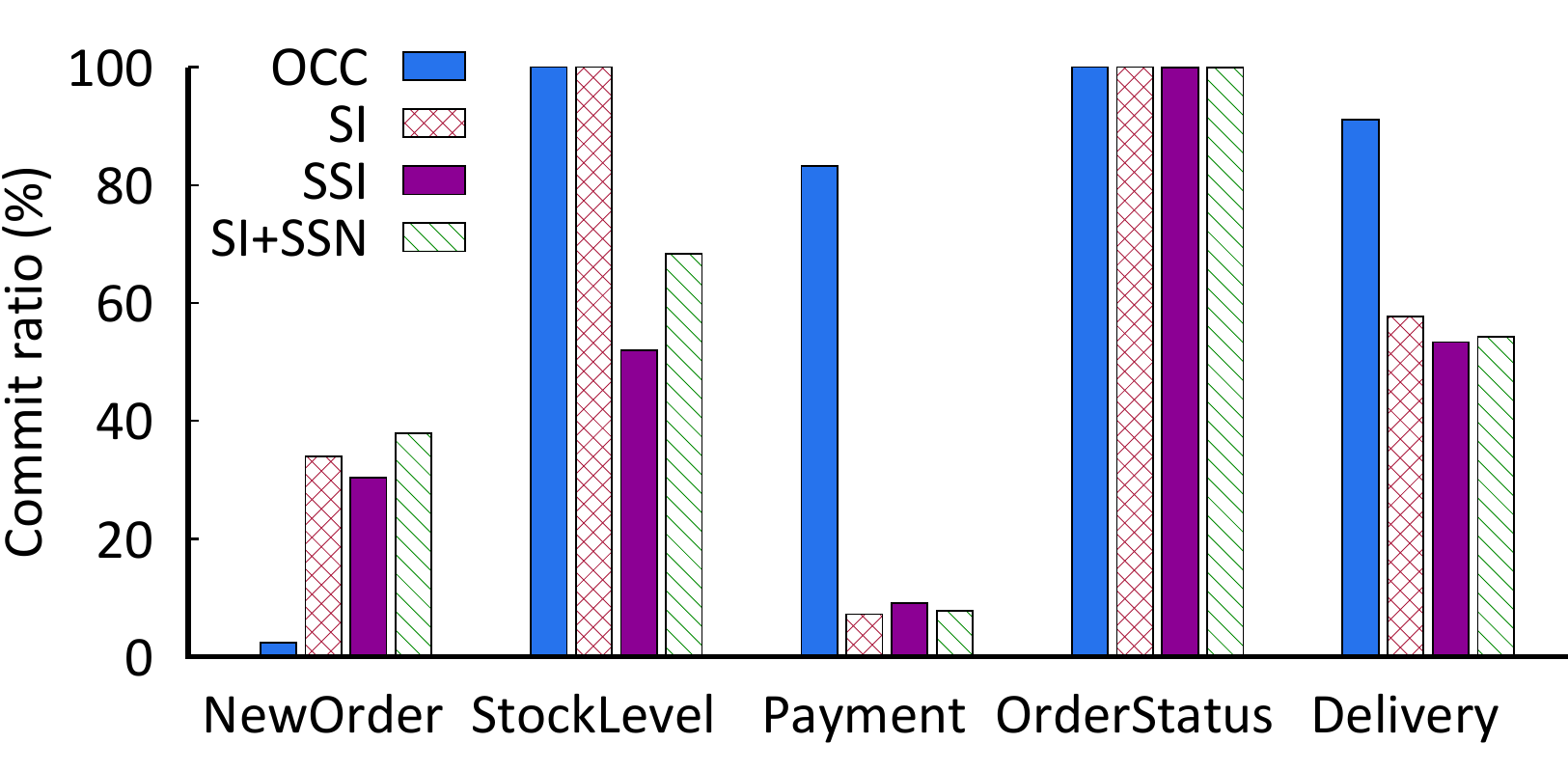}
\caption{Commit ratio of TPC-CC with 15 warehouses running at 60 threads when
aborted transactions are retried.}
\label{fig:tpcc-commit-ratio}
\end{figure}

\subsection{Safe retry and high-contention workloads}
\label{subsec:eval-retry}
This section evaluates SSN's safe retry property. We set both ERMIA and Silo to
retry aborted transactions until they commit successfully (therefore the
throughput breakdown---although not shown---will strictly follow the benchmark
specification). Fig.~\ref{fig:tpcc_contention-retry} shows the commit (left) and
abort (right) rates of TPC-CC with a varying number of concurrent worker threads.
Despite the extra effort needed to retry transactions, SSN matches the
performance of SSI and performs similarly to the case where we drop aborted
transactions (right side of Fig.~\ref{fig:tpcc}). OCC's commit rate collapsed as
core count increases, especially after 15 threads. Our profiling results show
that Silo spent the vast majority of CPU cycles (higher than 60\%) on retrying
index insertions (mostly for New-Order), minimizing the available cycles for
other transactions and getting little useful work done. Therefore, as shown by
Fig.~\ref{fig:tpcc_contention-retry}(right), OCC also kept a much lower abort
rate than SSI/SSN did. Compared to Silo, ERMIA uses indirection arrays for
multi-versioning~\cite{BwTree,ERMIA,IndArray}. This design makes tuple insertion
less reliant on index performance: Silo needs to first retry index insertion
before finalizing the tuple write at commit time, putting tremendous pressure on
the index, while ERMIA only needs to insert to the index after successfully
appended an entry in the table's indirection array, amortizing much contention on
the index. 

Although both SSI and SI+SSN commit and abort similar numbers of transactions,
they do so because of different reasons:  SSN exhibits higher accuracy and lower
abort rate due to certification failures, i.e., it aborts fewer transactions due
to serializability check than SSI does.
Fig.~\ref{fig:tpcc_contention-serial-abort} shows this effect. As we increase
the number of concurrent threads, SSI tends to abort more transactions due to
certification failures, while SI+SSN tends to abort more due to other reasons
including write-write conflicts and unsuccessful lookups. We also note that a
random backoff strategy during retries can largely mitigate this problem in Silo.
However, this will introduce a large variance in transaction latency, whereas
SSN's safe retry property can keep good performance and maintain stable latency.

In Fig.~\ref{fig:tpcc_contention-15wh-retry} we further stress the CC schemes
under high contention: the number of warehouses is fixed to 15 and we vary the
number of concurrent threads. Like previous experiments, each thread chooses a
random warehouse and retries a transaction until it is successfully committed. At
15 threads, all CC schemes exhibit exactly the same performance numbers as shown
by Fig.~\ref{fig:tpcc_contention-retry}. As we increase the number of concurrent
threads (i.e., more contention), as shown by
Fig.~\ref{fig:tpcc_contention-15wh-retry}(left), multi-version schemes can still
benefit from the increased parallelism. OCC, however, almost completely collapsed
as it keeps failing index insertions.
Fig.~\ref{fig:tpcc_contention-15wh-retry}(right) exhibits a similar but more
significant effect in Fig.~\ref{fig:tpcc_contention-serial-abort}, showing SSN's
accuracy and robustness under high contention. We also observed the similar trend
under TPC-CC without retrying aborted transactions under moderate levels of
contention (equal numbers of warehouses and concurrent threads). For example,
when running at 60 threads, SSN exhibits overall $\sim$60\% lower aborts due to
serialization failures when compared to SSI. 

Fig.~\ref{fig:tpcc-commit-ratio} shows the commit ratio for each transaction and
CC scheme running at 60 threads with 15 warehouses. The vertical axis represents
the percentage of a transaction committed out of all its retries. Similar to the
previous experiment, OCC has exhibited a bias toward the write-intensive Payment
transaction with a higher than 80\% of commit ratio, i.e., on average at most one
in five transactions needs to retry. But the percentage of committed NewOrder
transactions is only 2.36\%, due to its repetitive failure of index insertion.
For NewOrder, SI+SSN achieved the highest commit ratio, although SSI and SI+SSN
have similar aggregate throughput. Both SI and OCC have 100\% commit ratio for
the two read-only transactions (StockLevel and OrderStatus): the former does not
track and validate reads, while the latter uses a read-only snapshot for
read-only transactions. Finally, the commit ratios under SSI and SI+SSN for
StockLevel show the accuracy of SSN which achieves a 17\% higher commit ratio
when compared to SSI.

\section{Related work}
\label{sec:related-work}
An initial presentation of the SSN protocol (without some optimizations we report here) 
was in \cite{SSN-DAMON}.
This paper further extends SSN with more optimizations and functionality, 
such as a carefully designed parallel pre-commit protocol,
the adaptation of safe snapshot, phantom protection, 
optimizations for heavy-weight read-mostly transactions, etc.

The possibility of various isolation levels, some of which are not serializable,
was defined by Gray et al.~\cite{GrayLPT76}. Definitions of isolation properties
based on patterns of dependency edges are given by Adya~\cite{adya99}. Recent
efforts on CC focus on optimistic and multi-version based methods, such as OCC
defined by Kung et al.~\cite{KungR81} and SI defined academically (and proven
non-serializable) by Berenson et al.~\cite{BerensonBGMOO95}. 

Certification approaches that guarantee serializability are forms of optimistic
CC. Obtaining exact accuracy using serialization graph testing was proposed by
Casanova and Bernstein~\cite{CasanovaB80}, and extended to support
multi-versioning by Hadzilacos~\cite{Hadzilacos88}. A different approach tests
for cycles before transactions start in a real-time database system~\cite{ll00}. 
Recent OCC-based systems focus on eliminating physical contention to achieve
good performance. FOEDUS~\cite{FOEDUS} uses an extremely decentralized design
inspired by Silo~\cite{silo} to provide high performance on traditional TPC-C
like workloads. BCC~\cite{BCC} identifies specific dependency graph patterns to
reduce false aborts caused by vanilla OCC. Compared to SSN and multi-version
CC, these approaches cannot well support emerging heterogeneous workloads
featuring a significant amount of read-mostly transactions, as we have shown in
Sect.~\ref{subsec:eval-read-opt}.

The original SSI algorithm runs specifically along with SI to ensure serializable
executions~\cite{crf09}. An improved form (which we call SSI in our paper) was
implemented in PostgreSQL~\cite{pg12}. Revilak et al. proposed accompanying SI
with an exact certification using serialization graph testing~\cite{RevilakOO11}.
Other certification algorithms have been developed that can be used in a
snapshot-based system (one where all reads within a transaction come from a
common snapshot): Lomet et al.~\cite{lfw+12} choose a commit timestamp from an
allowed interval, and the chosen timestamp is the effective serial order commit
time. SSN, on the other hand, uses the commit time as timestamp, and tracks
excluded values. The chosen timestamp does not necessarily coincide with the
actual serial order commit time. Hekaton~\cite{hekaton} specifically aims for
main-memory stores and rejects all back edges. Deuteronomy~\cite{Deuteronomy}
separates physical and logical operations with dedicated data and transaction
components, and thus supports phantom protection while keeping a separation
between the CC and storage layers~\cite{Range-MVCC}. Neumann~\cite{HyperMVCC} et
al. adapts precision locking~\cite{PrecisionLocks} and uses undo buffers in
Hyper~\cite{hyper} to validate serializability.

Another class of proposals ensure serializable execution by doing static
pre-analysis of the application mix~\cite{flo+05,jfk+07,AlomariFR14}. In the
context of main-memory optimized systems, Bohm~\cite{Bohm} determines
serializable schedules prior to transaction execution, requiring that
transactions submitted in their entirety with the write sets deducible before
execution. Unlike SSN, these methods are not suitable with ad-hoc queries or data
dependent queries.

\section{Conclusions}
\label{sec:conclusion}

In this paper, we have presented the serial safety net (SSN), a cheap certifier
that can overlay a variety of concurrency control schemes and make them
serializable. We prove the correctness of SSN, we show how SSN can be efficiently
implemented for multi-version systems, and we have evaluated SSN in both
simulation and ERMIA, a recent main-memory multi-version database system designed
for modern hardware. 

SSN is robust against a variety of workloads, ranging from traditional OLTP
applications to emerging heterogeneous workloads.  In particular, we have
proposed specific optimizations for these heterogeneous workloads where long,
\textit{read-mostly} transactions are of paramount importance.  With the help of
a carefully designed lock-free, parallel commit protocol, SSN adds minimal
overhead to the underlying CC scheme in a multi-version system, in terms of
read/write tracking and commit-time validation.  Experiments using TPC-C and
TPC-E based benchmarks show that SSN is superior to prior state-of-the-art, in
being more accurate (fewer aborts and higher commit ratio), more general (not
requiring SI), more robust against retries and more friendly to emerging
heterogeneous workloads that features read-mostly transactions.

\section*{Acknowledgments}
The authors gratefully acknowledge Goetz Graefe and Harumi Kuno for their
invaluable suggestions and input they provided during the development of SSN, as
well as Kangnyeon Kim for his help building the prototype. We also thank the
anonymous reviewers for their comments and suggestions that greatly improved this
paper.

\bibliographystyle{spmpsci}      
\bibliography{citations}  

\end{document}